\documentclass{article}

\usepackage{graphicx} 
\usepackage{defs}
\usepackage{microtype}

\usepackage[maxnames=4]{biblatex}       
\title{Simulation Limitations of Affine Cellular Automata}

\usepackage{bigfoot}

\DeclareNewFootnote{AAffil}[arabic]
\DeclareNewFootnote{ANote}[fnsymbol]

\usepackage{etoolbox}
\usepackage{hyperref}

\makeatletter
\patchcmd\maketitle{\def\@makefnmark{\rlap{\@textsuperscript{\normalfont\@thefnmark}}}}{}{}{}
\makeatother

\makeatletter
\def\thanksAAffil#1{
  \footnotemarkAAffil\protected@xdef\@thanks{\@thanks%
        \protect\footnotetextAAffil[\the \c@footnoteAAffil]{#1}}%
}
\def\thanksANote#1{%
  \footnotemarkANote%
  \protected@xdef\@thanks{\@thanks%
        \protect\footnotetextANote[\the \c@footnoteANote]{#1}}%
}
\makeatother

\providecommand{\keywords}[1]
{
  \small	
  \textbf{\textit{Keywords: }} #1
  \normalsize
}

\author{%
  Barbora Hudcová%
  \thanksAAffil{Charles University, Faculty of Mathematics and Physics, Department of Algebra, Sokolovská~83, 18600 Prague}$^{,}$\thanksAAffil{CTU, Czech Institute of Informatics, Robotics and Cybernetics, Jugoslávských partyzánů~3, 160 00 Prague}$^{,}$%
  \thanksANote{Corresponding author; email: bara.hudcova@gmail.com, Charles University, Faculty of Mathematics and Physics, Department of Algebra, Sokolovská~83, 18600 Prague, Czech Republic}
  , %
  Jakub Krásenský%
  \footnotemarkAAffil[1]$^{,}$\thanksAAffil{Czech Technical University in Prague, Faculty of Information Technology, Thákurova~9, 160 00 Prague}%
}

\date{}

\begin{document}

\maketitle

\begin{abstract}
   Cellular automata are a famous model of computation, yet it is still a challenging task to assess the computational capacity of a given automaton; especially when it comes to showing negative results. In this paper, we focus on studying this problem via the notion of CA intrinsic simulation. We say that automaton $\A$ is simulated by $\B$ if each space-time diagram of $\A$ can be, after suitable transformations, reproduced by $\B$.
   
   We study affine automata -- i.e., automata whose local rules are affine mappings of vector spaces. This broad class contains the well-studied cases of linear automata. The main result of this paper shows that (almost) every automaton affine over a finite field $\F_p$ can only simulate affine automata over $\F_p$. We discuss how this general result implies, and widely surpasses, limitations of \chng{linear and} additive automata previously proved in the literature. 
   
   We provide a formalization of the simulation notions into algebraic language and discuss how this opens a new path to showing negative results about the computational power of cellular automata using deeper algebraic theorems.
\end{abstract}

\keywords{cellular automata, simulation capacity, affine cellular automata, grouping}

\section{Introduction}

Cellular Automata (CAs) are a famous model of computation which demonstrates that very simple local rules can produce complex behaviour. In general, they have a great potential to solve challenging tasks efficiently due to their massively parallel nature. However, it is a long-standing challenge to formally assess the computational capacity of a given CA.

The usual formal method of demonstrating a CA's computational power is to prove its Turing completeness. Typically, this is done by taking a classical universal model of computation (Turing machine, tag system, etc.) and embedding its computations into the space-time diagrams of the given CA. These embeddings serve as impressive demonstrations that cellular automata are indeed complex systems. Moreover, the embeddings serve as powerful tools for constructing very compact computationally universal systems. However, there are several drawbacks to this approach.
\begin{enumerate}
    \item Cellular automata are highly parallel systems, thus, it is ineffective to simulate a sequential Turing machine in its space-time diagrams to perform computations \cite{universalities_in_cas}.
    \item Given a CA, showing that it is Turing complete is very laborious. Most proofs of this sort rely on noticing localized structures in the CAs' space-time diagrams whose interactions are typically encoded as basic logic gates \cite{universality_in_eca, turing_universality_of_gol}.
    \item Despite several attempts \cite{universality_and_cellular_automata, decidability_and_universality_in_symbolic_dynamical_systems, the_game_of_life_universality_revisited}, there is no single generally accepted formal definition of simulation so far \cite{universalities_in_cas}. Thus, it is extremely difficult to prove convincing negative results. 
\end{enumerate}

A different approach to formally assessing a CA's computational capacity is through the notion of \emph{\chng{CA intrinsic simulation}}. Informally, we say that CA $\A$ is simulated by $\B$ if each space-time diagram of $\A$ can be, after suitable transformations, reproduced by $\B$. We argue that comparing two cellular automata is much more natural than comparing a CA with a Turing machine, since in the latter case, the architectures of the systems differ substantially. Past works have explored various notions of CA simulations, typically focusing on positive results: for a fixed family of CAs and a fixed CA simulation definition, authors construct \emph{intrinsically universal CAs}; i.e., cellular automata that are able to simulate any other CA within the fixed family \cite{simple_universal_ca, intrinsic_universality_of_reversible_1D_ca, reversible_spacetime_simulation_of_cas, the_quest_for_small_universal_cas, the_game_of_life_universality_revisited, intrinsic_universality_problem_of_1d_cas, communication_complexity_and_intrinsic_universality_in_ca}.

In contrast, the complementary work focuses on the negative results. For various notions of CA simulations, the goal is to show that particular types of automata are limited in terms of what they can simulate. Such results are scarce, yet they bring a valuable insight into the structure of the CA space imposed by the simulation relation. For certain CA simulation definitions, negative results have been shown about various classes of CAs such as nilpotent CAs or particular additive automata \chng{\cite{inducing_order_on_ca_by_grouping, additive_cas_over_Zp_bottom, bulking2}}. \chng{In \cite{communication_complexity_and_intrinsic_universality_in_ca} the authors show that certain CAs cannot be intrinsically universal using an interesting approach based on the notion of communication complexity.}

Generally, each CA intrinsic simulation definition considers a certain class of CA transformations $\mathrm{T}$ that map each automaton $\B$ into a class of related automata $\mathrm{T}(\B)$. Then, we say that $\A$ can be simulated by $\B$ if $\A \in \mathrm{T}(\B)$. We propose an informal classification of the previously studied transformations into:
\begin{itemize}
    \item algebraic: transformations of the CA's local rule; e.g.: products, sub-automata, quotients, iterations
    \item geometric: transformations of the CA's grid structure; e.g.: tiling of the grid space and grouping of multiple cells, shifts, reflections
\end{itemize}

In this work, we propose a definition of CA simulation that is, to the best of our knowledge, the most general algebraic one so far. We define the CA simulation and study its properties in Section \ref{section:defining_simulation}. The importance of this section lies in formalizing all the CA notions in abstract algebraic language. This will allow us to see connections to well-established algebraic fields that can provide powerful tools for analysing the CA simulation capacities. 

In Section \ref{section:affine_introduction}, we introduce the class of affine CAs -- automata whose local rules are affine mappings of vector spaces. This general class contains the much studied additive \chng{and linear} CAs \cite{algebraic_properties_of_ca, exact_results_for_determinstic_cas, linear_cas_and_recurring_sequences}. Therein, we present Theorem \ref{thm:main_result} -- the main result of this paper, paraphrased below:
\begin{blank}
  Let $\B$ be a cellular automaton affine over a finite field $\F_p$. If $\B$ satisfies a certain mild condition (``the outer components of its local rule are bijective''), then all CAs simulated by $\B$ are also affine over $\F_p$ (and satisfy the same condition).  
\end{blank}

\noindent Thus, we show that affine CAs have very limited computational capacity. We note that the mild condition does not impose any important limitations as it is satisfied by almost all previously studied, non-trivial cases of affine CAs. We also note that our main result already implies the negative results about additive CAs proved previously in the literature \cite{additive_cas_over_Zp_bottom, bulking2}. \chng{Further, we discuss that the proofs we present extend to other general automata classes such as the abelian CAs.} Even when applied to some of the best known additive automata, our results give considerably more than was known before. To give one concrete example, it yields that ECA 60 \cite{structure_of_eca_rule_space} cannot simulate any other elementary cellular automaton (up to isomorphism); this is discussed in Example \ref{eca_60}.

Section \ref{section:affine_limitations} contains elementary proofs of the results leading up to Theorem \ref{thm:main_result}. \chng{In Subsection \ref{subsection:generalizations}, we discuss how our proofs generalize and yield negative results about the simulation capacities of other CA classes, such as the abelian automata (and their affine counterparts).}
In Section \ref{section:conclusion}, we point towards a connection to deeper results from universal algebra. We believe that further study of connections between general algebra and CAs could be fruitful and provide insightful results about the computational limitations of various cellular automata classes.

\section{Defining Simulation of Cellular Automata} \label{section:defining_simulation}
\begin{defn}[Cellular automaton]
Let $S$ be a finite set, $r \in \N$, and $f: S^{2r+1} \rightarrow S$ a function. Let $S^\Z$ denote the space of all bi-infinite sequences $c =\cdots c_{-1} c_0 c_1 \cdots$, $c_i \in S$ for each $i \in \Z$. A \emph{one-dimensional cellular automaton (CA) with set of states $S$, radius $r$, and local rule $f$} is a dynamical system $\mathcal{A} = (S^\Z, F)$ where $F: S^\Z \rightarrow S^\Z$ is defined for any $c \in S^\Z$ and any position $i \in \Z$ as:
\begin{equation} \label{def:global_F}
    F(c)_i = f(c_{i-r}, \ldots, c_{i-1}, c_i, c_{i+1}, \ldots, c_{i+r}).
\end{equation}
We call $S^\Z$ the \emph{configuration space of $\A$} and $F$ its \emph{global rule}. We also call $(S^\Z, F)$ the \emph{global algebra}. Each 1D CA is determined by its \emph{local algebra} $\AA = (S, f)$. Sometimes, we also write $\AA = (S, f)_r$ to highlight the CA's radius. A \emph{space-time diagram} of $\A$ with initial configuration $c \in S^\Z$ and $t$ time-steps is a matrix whose rows are exactly $c, F(c), \ldots, F^t(c)$. Visualizations of CA space-time diagrams give valuable insights into the CA's dynamics. When depicting a space-time diagram, we only show a finite part of each row.
\end{defn}

In this paper, we focus on 1D cellular automata. Thus, whenever we talk about a CA, we implicitly mean a 1-dimensional one. The definitions as well as most of the results of this paper can be generalized to higher dimensions quite straightforwardly, but the technical details would be tedious. Before we proceed with the definitions of CA simulation, we briefly review some simple algebraic concepts.

\begin{mdframed}
Let $S$ be a set, $k \in \N$, and $f: S^k \rightarrow S$ a function. We call the tuple $\AA=(S, f)$ \emph{an algebra of type $k$}. Let $\BB=(T, g)$ be an algebra also of type $k$. We say that:
\begin{itemize}
    \item $\AA$ is \emph{isomorphic} to $\BB$ if there exists a bijection $\varphi: S \rightarrow T$ such that for each $s_1, \ldots, s_k \in S$: $$\varphi\big(f(s_1, \ldots, s_k)\big) = g\big(\varphi(s_1), \ldots, \varphi(s_k)\big).$$ We write $\AA \cong \BB$.
    \item $\AA$ is a \emph{subalgebra} of $\BB$ if $S \subseteq T$ and $\restr{g}{S^k} = f$.
    \item Let $\sim \, \subseteq T \times T$ be an equivalence relation. We call $\sim$ a \emph{congruence} on $\BB$ if:
    $$g(t_1, \ldots, t_k) \sim g(t'_1, \ldots, t'_k) \text{ whenever } t_1 \sim t'_1, \ldots, t_k \sim t'_k \text{ for any } t_1, \ldots, t_k, t'_1, \ldots, t'_k \in T.$$ 
    We denote by $[t]_\sim = \{t' \in T \mid t' \sim t \}$ the \emph{equivalence class} of $t \in T$. There is a well-defined algebra $\bigslant{\BB}{\sim} = (U, h)$ where $U = \{[t]_\sim \mid t \in T\}$ and $h: U^k \rightarrow U$ is defined as $h([t_1]_\sim, \ldots, [t_k]_\sim) \coloneqq [g(t_1, \ldots, t_k)]_\sim$. We say that $\AA$ is a \emph{quotient algebra} of $\BB$ if $\AA = \bigslant{\BB}{\sim}$ for some congruence $\sim$ on $\BB$.
    \item Let $\BB_1=(T_1, g_1), \ldots, \BB_n=(T_n, g_n)$ all be algebras of type $k$. We define their \emph{product} $\BB_1 \times \cdots \times \BB_n$ as the algebra $(T_1 \times \cdots \times T_n, h)$ where $h(\uu_1, \ldots, \uu_k)_i = g_i(u_1^i, \ldots, u_k^i)$ for each $1 \leq i \leq k$, $\uu_i \in T_1 \times \cdots \times T_n$, $\uu_i = u_i^1 \cdots u_i^n$.
\end{itemize}
\end{mdframed}

\subsection{CA Canonical Relations}
We start by defining canonical relations between CAs. We first do this in a purely algebraic manner and subsequently motivate the definitions by relating them to the CAs' global dynamics.

\begin{defn}[CA canonical relations] \label{def:basic_local_simulations}
Let $\A = (S^\Z, F)$ and $\B = (T^\Z, G)$ be CAs with local algebras $\AA=(S, f)_r$ and $\BB=(T, g)_r$ respectively. We say that:
\begin{itemize}
    \item \emph{$\A$ is automaton-isomorphic to $\B$} if $\AA$ is isomorphic to $\BB$.
    \item \emph{$\A$ is a sub-automaton of $\B$} if $\AA$ is a subalgebra of $\BB$.
    \item \emph{$\A$ is a quotient automaton of $\B$} if $\AA$ is a quotient algebra of $\BB$. 
    \item Let $\A_1, \ldots, \A_k$ be CAs with local algebras $\AA_1=(S_1, f_1)_r, \ldots, \AA_k=(S_k, f_k)_r$ respectively. We define their \emph{product} $\A_1 \times \cdots \times \A_k$ to be the CA given by the local algebra $\AA_1 \times \cdots \times \AA_k$.
\end{itemize}
It is natural to consider automata up to isomorphism. This gives rise to the natural algebraic operators we define below for a class $\KK$ consisting of local algebras of some family of cellular automata with radius $r \in \N$.
\begin{align*}
    \SSS(\KK)&=\{\AA \mid \exists \BB \in \KK \text{ and } \CC \text{ subalgebra of } \BB \text{ such that } \AA \cong \CC\}\\
    \HHH(\KK) &= \{\AA \mid \exists \BB \in \KK \text{ and } \CC \text{ quotient algebra of } \BB \text{ such that } \AA \cong \CC\}\\
    \Pfin(\KK) &= \{\AA \mid \exists \BB_1, \ldots, \BB_k \in \KK \text{ such that } \AA \cong \BB_1 \times \cdots \times \BB_k\}.
\end{align*}
The operators $\SSS, \HHH$, and $\Pfin$ transform a single local algebra $\AA$ into classes of algebras $\SSS(\AA)$, $\HHH(\AA)$, and $\Pfin(\AA)$. Subsequently, we will compose such operators and therefore, we generally define them on a class of local algebras $\KK$ rather than on a single local algebra $\AA$.
\end{defn}

The CA canonical relations of Definition \ref{def:basic_local_simulations} relate the local algebras of cellular automata. As such, they do not explicitly describe how this relationship translates to the CAs' global dynamics. The crucial observation is the following: Suppose that $\A$ is a CA isomorphic to $\B$, is a sub-automaton of $\B$ or is a quotient automaton of $\B$. Then, for any space-time diagram $\mathbf{c}$ of $\A$, there exists a space-time diagram $\mathbf{d}$ of $\B$ which can be ``translated'' to $\mathbf{c}$ via a very simple mapping. Thus, $\B$ can effectively reproduce any dynamics of $\A$. Before we formalize this observation, we introduce some further terminology.

Let $\varphi: S \rightarrow T$ be a mapping between finite sets. We define its \emph{canonical extension} $\overline{\varphi}: S^\Z \rightarrow T^\Z$ simply as $\overline{\varphi}(c)_i = \varphi(c_i)$ for each $c \in S^\Z$ and each $i \in \Z$.

\begin{obs} \label{obs:local_global_relation}
    Let $\A = (S^\Z, F)$ and $\B = (T^\Z, G)$ be CAs with local algebras $\AA=(S, f)_r$ and $\BB=(T, g)_r$ respectively. Then, it holds that:
    \begin{enumerate}
        \item \label{obs:iso} $\AA \cong \BB$ if and only if there exists a bijection $\varphi: S \rightarrow T$ such that $\ophi \circ F = G \circ \ophi$.
        \item \label{obs:sub} $\AA \in \SSS(\BB)$ if and only if there exists an injective mapping $\iota: S \rightarrow T$ such that $\oi \circ F = G \circ \oi$.
        \item \label{obs:quo} $\AA \in \HHH(\BB)$ if and only if there exists a surjective mapping $\pi: T \rightarrow S$ such that $F \circ \opi = \opi \circ G$.
    \end{enumerate}
\end{obs}
Figure \ref{fig:quotient_automaton} shows an example illustrating the notion of a quotient automaton.

\begin{ex} \label{ex:quotient_automaton} Let $\B = (T^\Z, G)$ be the CA with local algebra $\BB=(\{0,1,2,3 \}, g)_1$ defined as $g(x,y,z)= (x+z) \bmod 4$. We consider $\sim \, \subseteq \, \Z_4 \times \Z_4$ defined as $x_1 \sim x_2$ if and only if $x_1,x_2$ have the same parity. It is clear that whenever $x_1 \sim x_2$, $y_1 \sim y_2$, and $z_1 \sim z_2$ for any $x_1, x_2, y_1, y_2, z_1, z_2 \in \{0,1,2,3 \}$, then $g(x_1, y_1, z_1) \sim g(x_2, y_2, z_2)$. Thus, $\sim$ is a congruence on $\BB$ and we can study $\bigslant{\BB}{\sim}$.

Let $\A = (S^\Z, F)$ be the CA with local algebra $\AA=(\{0,1 \}, f)_1$ where $f(x,y,z)= (x+z) \bmod 2$ ($\A$ is the ECA 90). We define two mappings:
\begin{alignat*}
    \varphi \varphi: \quad &\bigslant{\BB}{\sim} \quad \longrightarrow \quad \AA  \quad \quad \quad \quad \quad \pi: \quad  && \BB \quad \longrightarrow \quad \AA\\
    & [0]_\sim = \{0, 2\}  \mapsto 0 && 0, 2  \mapsto 0\\
    & [1]_\sim = \{1, 3\}  \mapsto 1 && 1, 3  \mapsto 1
\end{alignat*}

It is straightforward to verify that $\varphi$ is an isomorphism between $\bigslant{\BB}{\sim}$ and $\AA$. Thus, it witnesses that $\AA \in \HHH(\BB)$. Moreover, one can easily check that the canonical extension of $\pi$ satisfies $F \circ \opi = \opi \circ G$ (thus, we have explicitly found the map $\pi$ from Observation \ref{obs:local_global_relation}, part \ref{obs:quo}.). Figure \ref{fig:quotient_automaton} illustrates how any space-time diagram of $\A$ can be obtained from a suitable space-time diagram of $\B$ using $\opi$ as the ``translation mapping''.
\end{ex}

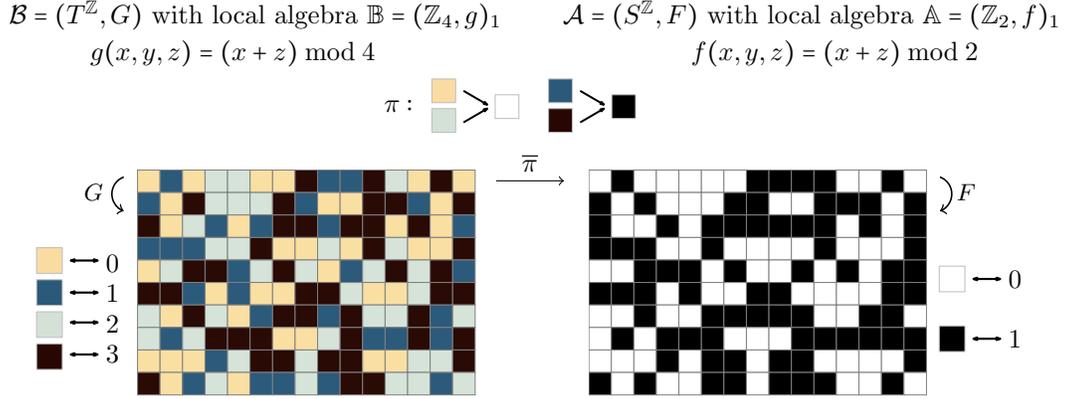
\begin{figure}
\centering
\begin{tikzpicture}
\node at (-.7, .4) {$\B = (T^\Z, G)$ with local algebra $\BB=(\Z_4, g)_1$};
\node at (-1, -.1) {$g(x,y,z)= (x+z) \bmod 4$};

\node at (6.7, .4) {$\A = (S^\Z, F)$ with local algebra $\AA=(\Z_2, f)_1$};
\node at (7, -.1) {$f(x,y,z)= (x+z) \bmod 2$};

\node at (1.2, -.8) {$\pi:$};
\node at (3, -.8) {\includegraphics[width=0.18\linewidth]{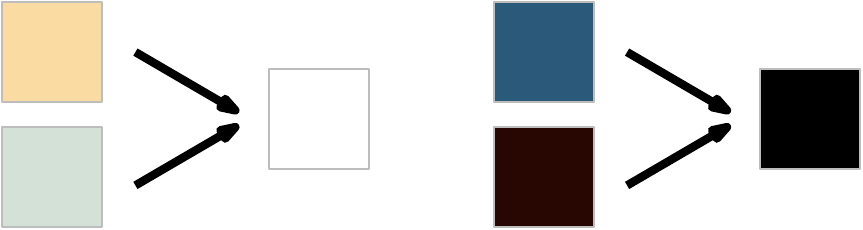}};

\node at (0, -3) {\includegraphics[width=0.3\linewidth]{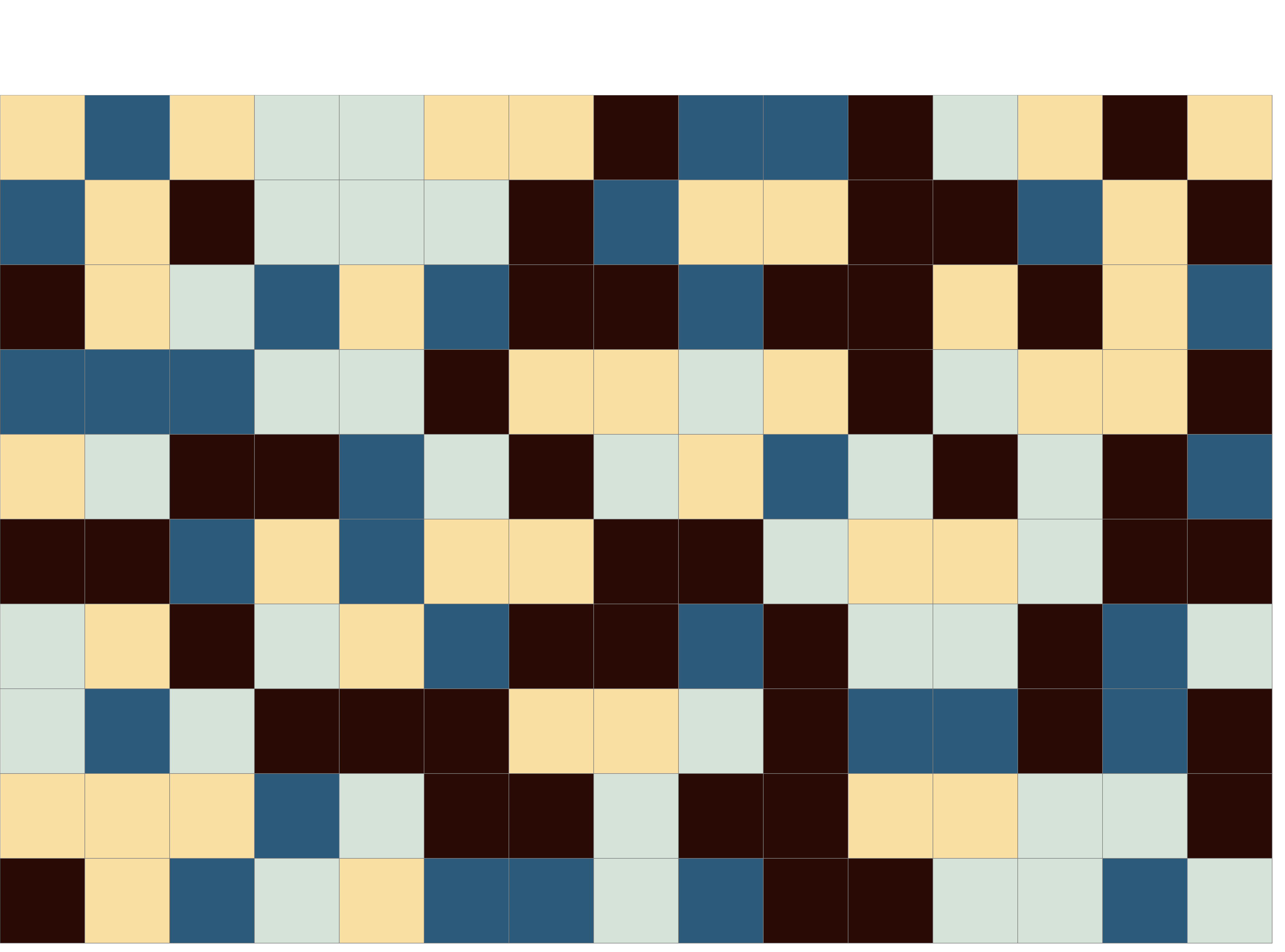}};
\node at (6, -3) {\includegraphics[width=0.3\linewidth]{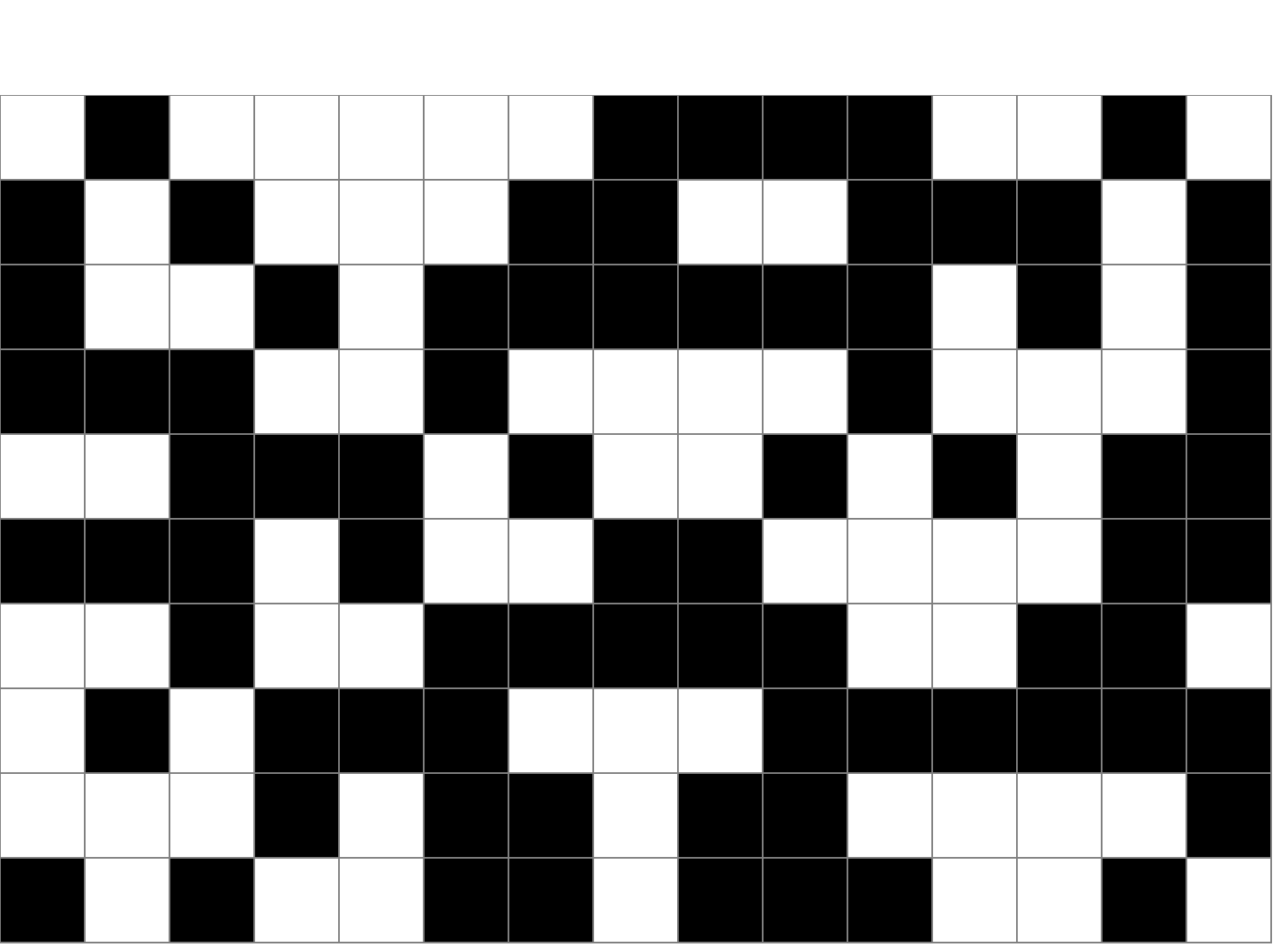}};

\node at (8.8, -3.5) (c2) {\includegraphics[width=0.055\linewidth]{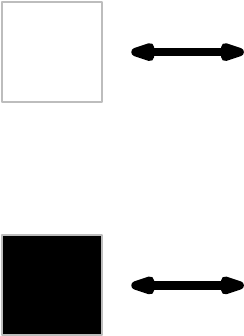}};
\node at (9.4, -3.11) {0}; \node at (9.4, -3.91) {1};
\node at (-3.2, -3.5) (c4) {\includegraphics[width=0.055\linewidth]{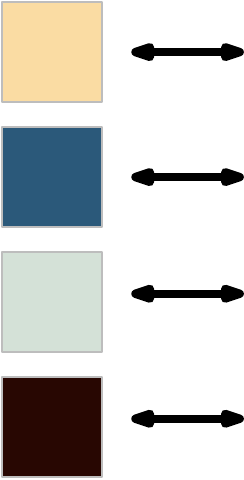}};
\node at (-2.6, -2.9) {0}; \node at (-2.6, -3.3) {1}; \node at (-2.6, -3.7) {2}; \node at (-2.6, -4.1) {3};

\path [->](2.5,-1.8) edge node [above,midway] {$\opi$} (3.4,-1.8);

\draw[->,semithick, color=black] (-2.45,-1.75) arc[radius=0.24, start angle=110, end angle=250]; \node[color=black] at (-2.85,-1.95){\small$G$};

\draw[->,semithick, color=black] (8.4,-1.75) arc[radius=0.24, start angle=70, end angle=-70]; \node[color=black] at (8.75,-1.95){\small$F$};
\end{tikzpicture}
\caption{Illustration of Example \ref{ex:quotient_automaton}. The figure shows that when $\AA \in \HHH(\BB)$, there exists a canonical extension $\opi$ which effectively ``translates'' space-time diagrams of $\B$ to any given diagram of $\A$.}
\label{fig:quotient_automaton}
\end{figure}

Let us fix the notation from Observation \ref{obs:local_global_relation}. 
The mappings $\ophi, \oi$, and $\opi$ provide means of ``translating'' between the space-time diagrams of $\A$ and $\B$. It is crucial that the mappings can be efficiently implemented by a computer program as they are extensions of mappings on finite sets. Moreover, the simplicity of the mappings guarantees that they do not process the information contained in the space-time diagrams in any non-trivial way. This is particularly important since, e.g., whenever $\AA \in \SSS(\BB)$, we would like to conclude that $\B$ is computationally stronger or equal to $\A$.

In contrast to the CA isomorphism, we can consider a general isomorphism between \chng{the CAs} $(S^\Z, F)$ and $(T^\Z, G)$ that can be witnessed by an arbitrary mapping $\psi: S^\Z \rightarrow T^\Z$ (e.g., a non-recursive one). In such a case, we say that $\A$ and $\B$ are \emph{isomorphic as dynamical systems}.

The operators $\SSS, \HHH$, and $\Pfin$ allow us to compare CAs' local rules. However, they do not take into account the most important aspect of cellular automata: the iterative application of the local rules. Thus, we describe below the established notion of \emph{CA iterative powers} also called \emph{grouping} \cite{inducing_order_on_ca_by_grouping}. This notion naturally extends the possible relationships we can study between two automata and, together with the $\SSS, \HHH$, and $\Pfin$ operators, will lead to the definition of CA simulation.

\subsection{Iterative Powers of CAs}
Let $\B = (T^\Z, G)$ be a CA with local algebra $\BB = (T, g)_r$. It is natural to iterate the global rule $G$ to obtain its powers $G^n$ for each $n \in \N$. This yields the ``iterated'' automaton $(T^\Z, G^n)$. The goal of this section is to describe the construction of the local algebra $\BB^{[n]}$ corresponding to the iterated automaton $(T^\Z, G^n)$ in such a way that the type of $\BB$ and $\BB^{[n]}$ remains the same.

We start by defining a function $\widetilde{g}$ which can be seen as an intermediate step between a CA's local map $g$ (with inputs of fixed length $2r+1$) and its global map $G$ (acting on infinite sequences).

\begin{defn}[Unravelling a local function]
 Let $T$ be a finite set and $g: T^{2r+1} \rightarrow T$ a function. We extend $g$ to a mapping $\widetilde{g}: \bigcup_{k=2r+1}^\infty T^k \rightarrow \bigcup_{k=1}^\infty T^k$ defined as:
\begin{align*}
\widetilde{g}(u_1 \cdots u_k) = g(u_1, \ldots, u_{2r+1}) g(u_2, \ldots, u_{2r+2}) \cdots g(u_{k-2r}, \ldots, u_{k})
\end{align*} for any $u_1 \cdots u_k \in T^k$, $k \geq 2r+1$. By $\widetilde{g}^n$ we simply mean the composition
$\widetilde{g}^n  = \underbrace{\widetilde{g} \circ \widetilde{g} \circ \cdots \circ \widetilde{g}}_{n \times}.$ Each application of $\widetilde{g}$ shortens the input sequence by $2r$. Therefore, the $n$-th iteration shortens the input by $2nr$ and $\widetilde{g}^n$ can be seen as a $(2nr{+}1)\text{-ary function}$ function into $T$, formally
$\restr{\widetilde{g}^n}{T^{2nr+1}}: T^{2nr+1} \rightarrow T$. This is illustrated in Figure \ref{fig:ftilde}.
\end{defn}
\begin{figure}[!htbp]
\centering
\begin{tikzpicture}
\def \x{0.4};

\draw[step=\x,color=black] (-6-2*\x,0) grid (-6+7*\x,\x);
\draw[orange, very thick] (-6-2*\x+0.03,.02) rectangle (-6-\x-0.02,\x-0.02);
\draw[orange, very thick] (-6-\x+0.03,.02) rectangle (-6-0.02,\x-0.02);
\draw[orange, very thick] (-6+0.03,.02) rectangle (-6+\x-0.02,\x-0.02);
\draw[orange, very thick] (-6+\x+0.03,.02) rectangle (-6+2*\x-0.02,\x-0.02);
\draw[orange, very thick] (-6+2*\x+0.03,.02) rectangle (-6+3*\x-0.02,\x-0.02);
\draw[orange, very thick] (-6+3*\x+0.03,.02) rectangle (-6+4*\x-0.02,\x-0.02);
\draw[orange, very thick] (-6+4*\x+0.03,.02) rectangle (-6+5*\x-0.02,\x-0.02);
\draw[orange, very thick] (-6+5*\x+0.03,.02) rectangle (-6+6*\x-0.02,\x-0.02);
\draw[orange, very thick] (-6+6*\x+0.03,.02) rectangle (-6+7*\x-0.02,\x-0.02);

\draw[step=0.4cm,color=black] (-6-\x,-.4) grid (-6+6*\x,0);
\draw[step=0.4cm,color=black] (-6,-.8) grid (-6+5*\x,-.4);
\draw[step=0.4cm,color=black] (-6+\x,-1.2) grid (-6+4*\x,-.8);
\draw[step=0.4cm,color=black] (-6+\x,-1.2) grid (-6+4*\x,-.8);
\draw[step=0.4cm,color=black] (-6+2*\x,-1.6) grid (-6+3*\x,-1.2);
\draw[orange, very thick] (-6+2*\x+0.03,-1.6+0.03) rectangle (-6+3*\x-0.02,-1.2-0.02);

\draw[->,semithick] (-6.8+3.8,.3) arc[radius=.2, start angle=70, end angle=-70];
\node at (-6.8+4.1, 0.1) {$\widetilde{g}$};

\draw[->,semithick] (-6.8+3.8,-.2) arc[radius=.2, start angle=70, end angle=-70];
\node at (-6.8+4.1, -.4) {$\widetilde{g}$};

\draw[->,semithick] (-6.8+3.8,-.7) arc[radius=.2, start angle=70, end angle=-70];
\node at (-6.8+4.1, -0.9) {$\widetilde{g}$};

\draw[->,semithick] (-6.8+3.8,-1.2) arc[radius=.2, start angle=70, end angle=-70];
\node at (-6.8+4.1, -1.4) {$\widetilde{g}$};

\draw[->,semithick, orange] (-2*\x-6.2,.2) arc[radius=.8, start angle=110, end angle=260];
\node at (-6.95-2*\x, -.6) {$\widetilde{g}^4$};

\end{tikzpicture}

\caption{Illustration of $\widetilde{g}$ and $\widetilde{g}^n$ for a ternary ($r=1$) function $g$ and $n=4$.}
\label{fig:ftilde}
\end{figure}
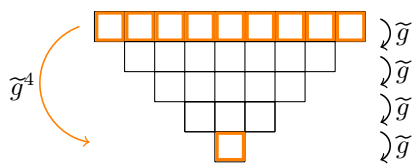

 It is straightforward to verify that the iterated CA $\left(T^\Z, G^n\right)$ has the local algebra $\left(T, \restr{\widetilde{g}^n}{T^{2nr+1}}\right)_{nr}$ with type $2nr+1$. 
Its type is different from $(T, g)_r$ whenever $n>1$. Our goal is to introduce a simulation relation between two automata via the notions of automata quotients, sub-automata, finite products, and iterative powers. \chng{Since the operators of sub-automata, quotients and finite products already preserve algebra types, it is natural to require this also from the iterative powers.} For them, this can be achieved by ``grouping together sequences of $n$ consecutive states''. This is formally defined below.

\begin{defn}
   Let $\B = (T^\Z, G)$ be a CA with local algebra $\BB = (T, g)_r$ and let $n \in \N$. We define the \emph{unpacking map} $o_n: (T^n)^\Z \rightarrow T^\Z$ for any configuration $c \in (T^n)^\Z$ and any $j \in \Z$, \chng{$j = nq + s$, $0 \leq s<n$}, as $\big(o_n(c)\big)_j = (c_q)_s$.
\end{defn}
 This is illustrated in Figure \ref{fig:unpacking_map}. 

\begin{figure}[htbp]
\centering
\begin{tikzpicture}
\def \x{0.25}
\def \lx{0.75}

\draw[xstep=\x,ystep=\x,color=black!40] (0,0) grid (6*\lx, \x);
\draw[xstep=\x,ystep=\x,color=black!40] (0,-3*\x) grid (6*\lx, -2*\x);

\draw[xstep=\lx,ystep=\x,color=black, very thick] (0,0) grid (6*\lx, \x);
\draw[xstep=\x,ystep=\x,color=black, very thick] (0,-3*\x) grid (6*\lx, -2*\x);

\draw[->,semithick] (6*\lx+0.1,\x/2-0.02) arc[radius=1.5*\x, start angle=70, end angle=-70]; \node at (6*\lx+0.7,-\x){\small$o_3$};
\end{tikzpicture}
\caption{Diagram of the unpacking map $o_3$.}
\label{fig:unpacking_map}
\end{figure}
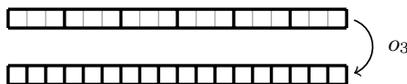

\begin{defn}
    Let $\B = (T^\Z, G)$ be a CA with local algebra $\BB = (T, g)_r$ and let $n \in \N$. We define the global map $G^{[n]}:\left(T^n\right)^\Z \rightarrow (T^n)^\Z$ as $G^{[n]} = o_n^{-1} \circ G^n \circ o_n$. This yields a new dynamical system $\left((T^n)^\Z, G^{[n]}\right)$. Further, we define the function $g^{[n]}: \underbrace{T^n \times \cdots \times T^n}_{(2r+1) \times} \rightarrow T^n$ as: 
$$g^{[n]}(\xx_{-r}, \ldots, \xx_r) \coloneqq \widetilde{g}^n (x_{-r}^1 \cdots x_{-r}^n x_{-r+1}^1 \cdots x_{-r+1}^n \cdots x_r^1 \cdots x_r^n) $$
where for each $i$, $\xx_i = x_i^1 \cdots x_i^n \in T^n$. The function $g^{[n]}$ has arity $2r+1$ and is illustrated in Figure \ref{fig:f_box}.
\end{defn}

\begin{obs}
Let $\B = (T^\Z, G)$ be a CA with local algebra $\BB = (T, g)_r$ and let $n \in \N$. Then, $\left((T^n)^\Z, G^{[n]}\right)$ is a CA with local algebra $(T^n, g^{[n]})$ and radius $r$.

Clearly, there is no automaton-isomorphism between $(T^\Z, G^n)$ and $\left((T^n)^\Z, G^{[n]}\right)$ since their state spaces have different sizes. However, they are isomorphic as dynamical systems via the unpacking mapping $o_n: (T^n)^\Z \rightarrow T^\Z$. Thus, the two systems do have equivalent dynamics.
\end{obs}

 \begin{defn}[Iterative power of an automaton] \label{def:iterative_automaton}
      Let $\B = (T^\Z, G)$ be a CA with local algebra $\BB = (T, g)_r$. For $n \in \N$, we define the $n$-th \emph{iterative power} of $\B$ as the automaton $\B^{[n]} = \left((T^n)^\Z, G^{[n]}\right)$ with local algebra $\BB^{[n]} = \left( T^n, g^{[n]} \right)_r$.
 \end{defn}

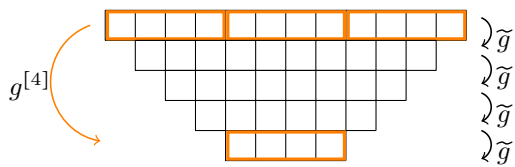
\begin{figure}[!htbp]
\centering
\begin{tikzpicture}
\def \x{0.4};

\draw[step=\x,color=black] (0,0) grid (12*\x,\x);
\draw[orange, very thick] (.03,.02) rectangle (4*\x-0.02,\x-0.02);
\draw[orange, very thick] (4*\x+.03,.02) rectangle (8*\x-0.02,\x-0.02);
\draw[orange, very thick] (8*\x+.03,.02) rectangle (12*\x-0.02,\x-0.02);

\draw[step=0.4cm,color=black] (\x,-.4) grid (11*\x,0);
\draw[step=0.4cm,color=black] (2*\x,-.8) grid (10*\x,-.4);
\draw[step=0.4cm,color=black] (3*\x,-1.2) grid (9*\x,-.8);
\draw[step=0.4cm,color=black] (4*\x,-1.6) grid (8*\x,-1.2);
\draw[orange, very thick] (4*\x+0.03,-1.58) rectangle (8*\x-0.02,-1.22);

\draw[->,semithick, orange] (-.2,.2) arc[radius=.8, start angle=110, end angle=260];
\node at (-1, -.6) {$g^{[4]}$};

\draw[->,semithick] (5,.3) arc[radius=.2, start angle=70, end angle=-70];
\node at (5.3, 0) {$\widetilde{g}$};

\draw[->,semithick] (5,-.2) arc[radius=.2, start angle=70, end angle=-70];
\node at (5.3, -.5) {$\widetilde{g}$};

\draw[->,semithick] (5,-.7) arc[radius=.2, start angle=70, end angle=-70];
\node at (5.3, -1) {$\widetilde{g}$};

\draw[->,semithick] (5,-1.2) arc[radius=.2, start angle=70, end angle=-70];
\node at (5.3, -1.5) {$\widetilde{g}$};

\end{tikzpicture}

\caption{Illustration of $g^{[4]}$ for a ternary function $g$.}
\label{fig:f_box}
\end{figure}

In this way, the local algebra $\BB^{[n]}$ describes exactly $n$ iterations of the CA's local rule while preserving the type of $\BB$. Therefore, each local algebra $\BB$ of a CA gives rise to a series of algebras of the same type, but operating on ``larger scales'' in both state space and time:

$$ \BB=(T, g)_r,\,\BB^{[2]}=(T^2, g^{[2]})_r,\,\BB^{[3]}=(T^3, g^{[3]})_r, \, \ldots$$

\begin{defn} \label{def:iterative_powers}
Let $\KK$ be a class of local algebras of a family of CAs with radius $r$. We define the \emph{iterative power} operator:
\begin{align*}
    \Xi(\KK) &= \{\AA \mid \exists \BB \in \KK \text{ and } n \in \N \text{ such that } \AA \cong \BB^{[n]}\}.\\
\end{align*}
\end{defn}

In contrast to the operators $\HHH, \SSS$, and $\Pfin$ that are algebraic in their nature, $\Xi$ can be viewed as a geometrical transformation of a CA's space-time diagrams. Indeed, the $n$-th iterative power essentially ``selects'' every $n$-th row of the space-time diagram, and groups the ``cells'' together into packages of $n$ consecutive cells via the $o_n^{-1}$ operator. This grouping is necessary for preserving the iterated automata's radius. We can finally proceed to introducing the general definition of CA simulation.

\begin{defn}[CA simulation] \label{def:ca_simulation}  Let $\A$ and $\B$ be cellular automata with local algebras $\AA = (S, f)_r$ and $\BB = (T, g)_r$ respectively. We say that $\B$ \emph{simulates} $\A$ if $\AA \in \HHH\SSS\Pfin \Xi(\BB)$; and we write $\AA \preceq \BB$ or $\A \preceq \B$.  
\end{defn}

Note that the simulation relation is well-defined only for cellular automata with the same radius. However, whenever we consider an automaton $\A$ with radius $r$, we can also interpret it as a CA with radius $r' \in \N$ for any $r' > r$ in a natural way. This allows us to compare automata with different radii after all.

A perhaps more natural definition of simulation would be to require the following properties: First, $\A \preceq \B$ if at least one of the following properties hold: 
\begin{enumerate}
    \item $\AA \in \SSS(\BB)$; $\A$ is (isomorphic to) a sub-automaton of $\B$.
    \item $\AA \in \HHH(\BB)$; $\A$ is (isomorphic to) a quotient automaton of $\B$.
    \item $\AA \in \Xi(\BB)$; $\A$ is (isomorphic to) an iterative power of $\B$.
    \item $\AA \in \Pfin\big(\HHH(\BB) \cup \SSS(\BB) \cup \Xi(\BB)\big)$; $\A$ is (isomorphic to) a finite product of automata which are themselves ``simulated'' by $\B$.
\end{enumerate}

This is well-motivated, since in all these cases, any computation in $\A$ can be easily recovered from a computation in $\B$. But any reasonable definition of simulation also must be transitive. \chng{In the following subsection, we show that we have chosen the operators $\HHH, \SSS, \Pfin, \Xi$ in Definition \ref{def:ca_simulation} carefully, exactly in this order, because this yields the transitive closure of the requirements above.}

\subsection{Elementary Properties of CA Simulation}

\begin{lemma} \label{lemma:Pfin_Pfin}
    Let $\AA$ be a local algebra of a CA $\A$ and let $m, n \in \N$. Then, $(\AA^{[m]})^{[n]} \cong \AA^{[mn]}$.
\end{lemma}
\begin{proof}
    Let $\A = (S^\Z, F)$. Then by definition $\A^{[mn]}= \left((S^{mn})^\Z, F^{[mn]}\right)$; $F^{[mn]} = o_{mn}^{-1} \circ F^{mn} \circ o_{mn}$, where $o_{mn}: (S^{mn})^\Z \rightarrow S^\Z$ is the unpacking map. Similarly, $(\A^{[m]})^{[n]} = \left( ((S^m)^n)^\Z, (F^{[m]})^{[n]} \right)$; $(F^{[m]})^{[n]} = o_n^{-1} \circ \left( o_m^{-1} \circ F^m \circ o_m \right)^n \circ o_n$, where $o_m: S^m \rightarrow S$ and $o_n: (S^m)^n \rightarrow S^m$. Let $\varphi: S^{mn} \rightarrow (S^m)^n$ be the natural bijection. We will show that its extension $\ophi$ satisfies $\ophi^{-1} \circ (F^{[m]})^{[n]} \circ \ophi = F^{[mn]}$. Then, Observation \ref{obs:local_global_relation} already implies that the corresponding local algebras are isomorphic. It is easy to verify that $o_m \circ o_n \circ \ophi = o_{mn}$. Then:
    \begin{align*}
        \ophi^{-1} \circ (F^{[m]})^{[n]} \circ \ophi &= \ophi^{-1} \circ o_n^{-1} \circ \left( o_m^{-1} \circ F^m \circ o_m \right)^n \circ o_n \circ \ophi\\
        &= \ophi^{-1} \circ o_n^{-1} \circ o_m^{-1} \circ F^{mn} \circ o_m \circ o_n \circ \ophi\\
        &= o_{mn}^{-1} \circ F^{mn} \circ o_{mn} = F^{[mn]}.
        \qedhere
    \end{align*}

\end{proof}

\begin{lemma} \label{lemma:Pfin_Xsi}
    Let $\AA_1, \ldots, \AA_k$ be local algebras of some CAs with radius $r$ and let $n \in \N$. Then $\left(\AA_1 \times \cdots \times \AA_k \right)^{[n]} \cong \AA_1^{[n]} \times \cdots \times \AA_k^{[n]}$.
\end{lemma}
\begin{proof}
    Let $\AA_i = (S_i, f_i)_r$ for $1 \leq i \leq k$.
    Similarly to Lemma \ref{lemma:Pfin_Pfin}, it is a straightforward, yet technical verification that the natural bijection $\varphi: (S_1 \times \cdots \times S_k)^n \rightarrow S_1^n \times \cdots \times S_k^n$ is an isomorphism witnessing that $\left(\AA_1 \times \cdots \times \AA_k \right)^{[n]} \cong \AA_1^{[n]} \times \cdots \times \AA_k^{[n]}$. 
\end{proof}

\begin{prop} \label{prop:hspx_is_closed}
    Consider any class of CAs with a fixed radius $r$ and let $\KK$ be the class of their local algebras. Then, $\HHH\SSS\Pfin \Xi(\KK)$ is already closed under the operators $\HHH, \SSS, \Pfin$, and $\Xi$. 
\end{prop} 
\begin{proof}
Let $\KK'$ be an arbitrary class of local algebras of the same type. It is a well-known result from universal algebra \cite{ua_bergman} that:
\begin{alignat*}{5}
    \HHH\HHH(\KK')&=\HHH(\KK'), \quad \quad \quad &&\SSS\SSS(\KK')&&=\SSS(\KK'), \quad \quad \quad \Pfin &&\Pfin(\KK')&&=\Pfin(\KK'),\\ 
    \SSS\HHH(\KK')&\subseteq \HHH\SSS(\KK'), &&\Pfin \HHH(\KK') &&\subseteq \HHH \Pfin (\KK'), &&\Pfin \SSS(\KK') &&\subseteq \SSS \Pfin(\KK').
\end{alignat*}
Moreover, Lemma \ref{lemma:Pfin_Pfin} implies $\Xi \Xi (\KK') = \Xi(\KK')$ and Lemma \ref{lemma:Pfin_Xsi} yields $\Xi \Pfin(\KK') \subseteq \Pfin \Xi(\KK')$. Further, in \cite[Lemma 3]{inducing_order_on_ca_by_grouping} the authors showed: If $\AA \in \SSS(\BB)$ for any local algebras $\AA$, $\BB$ of the same type, then $\AA^{[n]} \in \SSS(\BB^{[n]})$ for every $n\in\N$. Thus, $\Xi \SSS (\KK') \subseteq \SSS \Xi (\KK')$. In \cite{bulking2}, the authors outline the proof of $\Xi \HHH(\KK') \subseteq \HHH\Xi(\KK')$. We give an explicit proof here. 

Let $\AA = (S, f)_r$ and $\BB=(T, g)_r$ be such that $\AA \in \HHH(\BB)$ and let $n \in \N$. We show that $\AA^{[n]} \in \HHH(\BB^{[n]})$. Since $\AA \in \HHH(\BB)$, there is some surjective homomorphism $\varphi: T \twoheadrightarrow S$ of the local algebras. We extend it to a surjective mapping $\psi: T^n \twoheadrightarrow S^n$ simply by putting $\psi(\xx)_i = \varphi(x_i)$ for any $\xx \in T^n$ and $i \in \{1, \ldots, n\}$. Now it remains to verify that $\psi: \AA^{[n]} \rightarrow \BB^{[n]}$ is a surjective homomorphism of the algebras. First, let $m \geq 2r+1$ and $x_1, \ldots, x_m \in T$. We show that $\widetilde{f}(\widetilde{\varphi}(x_1\cdots x_m))=\widetilde{\varphi}( \widetilde{g}(x_1\cdots x_m))$. Clearly, for every $i \in \N$, it holds that $\widetilde{f}(\widetilde{\varphi}(x_1 \cdots x_m))_i = f(\varphi(x_{i-r}) \cdots \varphi(x_{i+r})) = \varphi(g(x_{i-r}, \ldots, x_{i+r})) = \widetilde{\varphi} (\widetilde{g}(x_1\cdots x_m))_i$. Further, by induction we obtain that for every $n \in \N$ and every sequence $x_1, \ldots, x_m \in T$, $m \geq 2nr+1$, it holds that $\widetilde{f}^n(\widetilde{\varphi}(x_1 \cdots x_m))=\widetilde{\varphi} (\widetilde{g}^n(x_1\cdots x_m))$. Next, let $\xx_{-r}, \ldots, \xx_r \in T^n$. We have:
$$f^{[n]}\big(\psi(\xx_{-r}), \ldots, \psi(\xx_r)\big) = \widetilde{f}^n\big(\widetilde{\varphi}(\xx_{-r}\cdots \xx_r)\big) = \widetilde{\varphi}\big(\widetilde{g}^n(\xx_{-r}\cdots \xx_r)\big) = \psi\big(g^{[n]}(\xx_{-r}, \ldots, \xx_r)\big).$$

Thus, indeed, $\psi: T^n \twoheadrightarrow S^n$ is a surjective homomorphism of algebras $\AA^{[n]}$ and $\BB^{[n]}$, and $\AA^{[n]} \in \HHH(\BB^{[n]})$. This already implies that $\Xi \HHH(\KK') \subseteq \HHH\Xi(\KK')$.

Now we are ready to finish the proof that $\HHH\SSS \Pfin \Xi(\KK)$ is already closed under $\HHH, \SSS, \Pfin$, and $\Xi$. We show this for the operator $\Xi$:
$$\Xi \HHH\SSS\Pfin \Xi(\KK) \subseteq \HHH \Xi \SSS \Pfin \Xi(\KK) \subseteq \HHH\SSS \Xi \Pfin \Xi (\KK) \subseteq \HHH\SSS \Pfin \Xi \Xi(\KK) = \HHH\SSS \Pfin \Xi(\KK)$$
where the first inclusion uses that $\Xi \HHH(\KK') \subseteq \HHH \Xi(\KK')$ for $\KK' = \SSS\Pfin \Xi(\KK)$ and the subsequent inclusions follow similar logic. For the other operators, the proof is analogous and takes less than one line, using the results summarized above. 
\end{proof}

\begin{cor}
    The CA simulation relation $\preceq$ is reflexive and transitive; i.e., it forms a preorder.
\end{cor}
\begin{proof}
    The reflexivity of $\preceq$ is clear. We briefly discuss the transitivity. Let $\AA$, $\BB$, $\CC$ be local algebras of CAs all with radius $r \in \N$ such that $\AA \preceq \BB$ and $\BB \preceq \CC$. By definition, $\AA \in \HHH \SSS \Pfin \Xi (\BB)$ and $\BB \in \HHH \SSS \Pfin \Xi(\CC)$. Thus, $\AA \in \HHH \SSS \Pfin \Xi (\BB) \subseteq \HHH \SSS \Pfin \Xi \HHH \SSS \Pfin \Xi (\CC) = \HHH \SSS \Pfin \Xi (\CC)$ where the last equality holds due to Proposition \ref{prop:hspx_is_closed}. Hence, $\AA \preceq \CC$. 
\end{proof}

\begin{defn}
    Let $r \in \N$. We define $\SING_r$ to be the class of local algebras of the form $(S,f)_r$ where $|S|=1$. 
\end{defn}

Clearly, any two local algebras in $\SING_r$ are isomorphic. Moreover, let $\AA \in \SING_r$, and consider any local algebra $\BB = (S, f)_r$ of a CA. Then, $\AA \preceq \BB$ since $\AA \in \HHH(\BB)$. \chng{Lastly, it is clear that $\AA$ can only simulate elements of $\SING_r$.} Thus, $\SING_r$ forms the minimum element within the class of all 1D CAs with radius $r$.

We note that our definition of CA simulation combines all classical algebraic operators that preserve the finiteness of the algebras. \chng{Compared to previous approaches, we have added the operator $\Pfin$ to the definition. We highlight that the negative results in this paper also hold (as a trivial corollary) when considering CA simulation based just on the operators $\HHH, \SSS$ and $\Xi$. However, adding $\Pfin$ has various advantages, which emphasize the connection to the field of universal algebra: For any class of local algebras $\KK$ of CAs with radius $r$, $\HHH\SSS \Pfin \Xi(\KK)$ forms a so-called \emph{pseudovariety}. Thus, as shown in \cite{on_pseudovarieties}, $\HHH\SSS \Pfin \Xi(\KK)$ can be characterized by a sequence of equations. Finding such equation sequences for different classes of CAs is a challenging problem that can provide important insight into the structure of the simulation relation and is an interesting line of future work. Another advantage is that with the $\Pfin$ operator, the simulation relation $\preceq$ forms an upper semi-lattice as we show in the following proposition.}

\chng{
\begin{prop}
    Let $r \in \N$ be a fixed CA radius. For $\AA$, $\BB$ local algebras of CAs with radius $r$, we say that $\AA$ and $\BB$ are equivalent if $\AA \preceq \BB$ and $\BB \preceq \AA$. Let us consider the set $\mathrm{CA^*}_r$ consisting of all local algebras of CAs with radius $r$ up to the equivalence. Then, $(\mathrm{CA^*}_r, \preceq)$ forms an upper semi-lattice.
\end{prop}
\begin{proof}
    Let $\AA, \BB \in \mathrm{CA^*}_r$. We will show that $\AA \times \BB$ is the supremum of $\AA$ and $\BB$; i.e., that $\AA \preceq \AA \times \BB$, $\BB \preceq \AA \times \BB$ and that whenever there is a $\CC \in \mathrm{CA^*}_r$ such that $\AA \preceq \CC$, $\BB \preceq \CC$, then $\AA \times \BB \preceq \CC$. This will show that $(\mathrm{CA^*}_r, \preceq)$ is an upper semi-lattice.

    The fact that $\AA \preceq \AA \times \BB$ and $\BB \preceq \AA \times \BB$ is clear. Let us consider $\CC \in \mathrm{CA^*}_r$ such that $\AA \preceq \CC$, $\BB \preceq \CC$. Then there exists $\CC_1 \in \SSS \Pfin \Xi(\CC)$ such that $\AA \in \HHH(\CC_1)$. Similarly, there exists $\CC_2 \in \SSS \Pfin \Xi(\CC)$ such that $\BB \in \HHH(\CC_2)$. Then clearly $\AA \times \BB \in \HHH(\CC_1 \times \CC_2)$ and $\CC_1 \times \CC_2 \in \Pfin \SSS \Pfin \Xi(\CC) = \SSS \Pfin \Xi(\CC)$. Thus, $\AA \times \BB \in \HHH \SSS \Pfin \Xi (\CC)$ and $\AA \times \BB \preceq \CC$.
\end{proof}
Note that in \cite{bulking2}, the simulation relation is considered without the $\Pfin$ operator yet with more general geometrical transformations. In such a case, the authors show that their definition of CA simulation induces neither an upper, nor a lower semi-lattice.
}

\chng{\subsection{Bi-permutivity}}
\chng{
To finish this section, we introduce a largely studied class of CAs called bi-permutive and show a simple (well-known) negative result: bi-permutive CAs can only simulate other bi-permutive CAs. 
We first define the notions ``permutivity'' and ``insensitivity''. Informally, a CA with local rule $f$ is $i$-permutive if its $i$-th coordinate is a bijection. In contrast, the CA is insensitive to its $i$-th input if the output of $f$ does not depend on the $i$-th coordinate.
\vspace{3mm}
\begin{defn}[permutivity, insensitivity] \label{def:permutivity_insensitivity}
  Let $\A$ be a CA with local algebra $\AA=(S, f)_r$ and let $-r \leq i \leq r$. We say that $\AA$ (or $\A$) is $i$-permutive if $f(s_{-r}, \ldots, s_{i-1},~\cdot~, s_{i+1}, \ldots, s_r): S \rightarrow S \text{ is a bijection for all } s_{-r}, \ldots, s_{i-1}, s_{i+1}, \ldots, s_r \in S.$ 
  We say that $f$ is insensitive to its $i$-th input if
  $f(s_{-r}, \ldots, s_{i-1},~\cdot~, s_{i+1}, \ldots, s_r): S \rightarrow S \text{ is a constant mapping for all } s_{-r}, \ldots, s_{i-1}, s_{i+1}, \ldots, s_r \in S.$
\end{defn}

Suppose a CA with local algebra $\AA=(S, f)_r$ is insensitive to its $-r$-th and $r$-th input. Then, we can define a CA with smaller radius $r-1$ that has equivalent dynamics. 

Many of the results presented in this paper hold for affine CAs with at least two permutive coordinates. A special case of such CAs are the \emph{bi-permutive CAs} that are well studied in the literature \cite{topological_dynamics_of_cellular_automata}. For a bi-permutive CA, one requires the permutive coordinates to be the ``left-most'' and ``right-most'' ones. Below, we define the bi-permutivity for cellular automata in a slightly more general way compared to the classical definition to account for the fact that we only study cellular automata with a symmetrical neighbourhood given by a particular radius.

\begin{defn}[bijective condition, bi-permutivity] \label{def:bijectivity}
Let $\AA=(S,f)_r$ be a local algebra of a CA with radius $r$. We say that $\AA$ satisfies the \emph{bijective condition (witnessed by $i$ and $j$)} if there exist $-r \leq i < j \leq r$ such that $\AA$ is $i$-permutive and $j$-permutive.

If $\AA$ further satisfies that for all $k<i$ and $k>j$, $\AA$ is $k$-insensitive, we say that $\AA$ is \emph{bi-permutive}, or \emph{$(i,j)$-bi-permutive}. In such a case, we write that $\AA \in \BPR_{r;i,j}$.
\end{defn}

An important property of bi-permutivity is that this property is preserved by the CA iterative powers, as stated in the following lemma.

\begin{lemma} \label{lemma:ksi_permutivity}
    Let $r \in \N$ and $-r \leq i < j \leq r$. Let $\AA$ be a local algebra of an $(i,j)$-bi-permutive CA. Then, for each $n \in \N$, $\AA^{[n]}$ is also $(i,j)$-bi-permutive.
\end{lemma}
\begin{proof}
 This is a straightforward generalization of Lemma 3 from \cite{additive_cas_over_Zp_bottom}.  
\end{proof}

We note that Lemma \ref{lemma:ksi_permutivity} does not in general hold for CAs satisfying the bijective condition since the bijectivity of the ``inner local rule's coordinates'' need not be preserved by the CA's iterated powers. This is illustrated in the following example.

\begin{ex}\label{ex:counterexample_bijective_condition}
    Let $\AA = (\F_2^2, f)_2$ where $\F_2$ is the finite field with two elements and $f(\xx_{-2}, \xx_{-1}, \xx_0, \xx_1, \xx_2) = A_{-2} \xx_{-2} + A_{-1} \xx_{-1} + A_{0} \xx_{0} + A_{1} \xx_{1} + A_{2} \xx_{2}$ where:
    \begin{equation*} \label{ex:Ai_matrices}
     \equalto{A_{-2}}{\begin{pmatrix}
0 	& 1 \\
0 	& 0 \\
\end{pmatrix}}  \quad 
\equalto{A_{-1}}{\begin{pmatrix}
1 	& 1 \\
0 	& 1 \\
\end{pmatrix}}  \quad 
\equalto{A_{0}}{\begin{pmatrix}
1 	& 0 \\
1 	& 1 \\
\end{pmatrix}}  \quad 
 \equalto{A_{1}}{\begin{pmatrix}
0 		& 0 \\
1 		& 0 \\
\end{pmatrix}}  \quad 
 \equalto{A_{2}}{\begin{pmatrix}
0 		& 0 \\
0 		& 0 \\
\end{pmatrix}.}
\end{equation*}
It is easy to see that for $-r \leq i \leq r$, $\AA$ is $i$-permutive if and only if $A_i$ is an invertible matrix. Hence, $\AA$ is $-1$-permutive and $0$-permutive and thus, satisfies the bijective condition. However, $\AA$ is not bi-permutive since $\AA$ is $2$-insensitive but ``is sensitive'' to all other coordinates. It is straightforward to verify that $\AA^{[2]}=(\F_2^4, f^{[2]})_2$ with $f^{[2]}(\xx_{-2}, \xx_{-1}, \xx_0, \xx_1, \xx_2)=\vo$. Thus, $\AA^{[2]}$ is insensitive to every coordinate and clearly does not satisfy the bijective condition.
\end{ex}

We can now proceed to showing the simple negative result that bi-permutive CAs can only simulate bi-permutive CAs.

\begin{prop}
    Let $r \in \N$ and $-r \leq i < j \leq r$. Let $\AA, \BB$ be local algebras of cellular automata with radius $r$ such that $\BB$ is $(i,j)$-bi-permutive. If $\AA \preceq \BB$ then $\AA$ is $(i,j)$-bi-permutive.
\end{prop}
\begin{proof}
    From Lemma \ref{lemma:ksi_permutivity} we have that $\Xi(\BB) \subseteq \BPR_{r;i,j}$. Next, let $\BB_1, \ldots, \BB_k \in \BPR_{r;i,j}$. Then, clearly also $\BB_1 \times \cdots \times \BB_k \in \BPR_{r;i,j}$. Thus, $\Pfin\Xi(\BB) \subseteq \BPR_{r;i,j}$.

    Next, suppose that $\BB' \in \BPR_{r;i,j}$, $\BB' = (S', f')_r$ and let $\AA=(S, f)_r$ be a sub-automaton of $\BB'$ (so $S \subseteq S'$). We will show that $\AA \in \BPR_{r;i,j}$. Let $s_{-r}, \ldots, s_{i-1}, s_{i+1}, \ldots, s_r \in S$. Then the map $f'(s_{-r}, \ldots, s_{i-1}, \cdot, s_{i+1}, \ldots, s_r): S' \rightarrow S'$ is a bijection and hence, its restriction to $S$ is injective and thus, also bijective since $S$ is finite. Hence, $\AA$ is $i$-permutive. Analogously, we show that $\AA$ is $j$-permutive. Since insensitivity to the suitable coordinates is clear, we have that $\AA \in \BPR_{r;i,j}$ and therefore, $\SSS\Pfin\Xi(\BB) \subseteq \BPR_{r;i,j}$. 

    Similarly, suppose that $\BB' \in \BPR_{r;i,j}$, $\BB' = (S', f')_r$ and let $\AA=(S, f)_r$ be a quotient automaton of $\BB'$; $\AA = \bigslant{\BB'}{\sim}$. We will show that $\AA \in \BPR_{r;i,j}$. Let $[s_{-r}]_\sim, \ldots, [s_{i-1}]_\sim, [s_{i+1}]_\sim, \ldots, [s_r]_\sim \in S$. Then the map $f'(s_{-r}, \ldots, s_{i-1}, \cdot, s_{i+1}, \ldots, s_r): S' \rightarrow S'$ is a bijection and hence, is surjective. Thus, $f([s_{-r}]_\sim, \ldots, [s_{i-1}]_\sim, \cdot, [s_{i+1}]_\sim, \ldots, [s_r]_\sim): S \rightarrow S$ is also surjective, and since $S$ is finite, also a bijection. Thus, $\AA$ is $i$-permutive. Analogously, we show that $\AA$ is $j$-permutive. Thus, it is easy to see that $\AA \in \BPR_{r;i,j}$. Therefore, $\HHH\SSS\Pfin\Xi(\BB) \subseteq \BPR_{r;i,j}$ which concludes the proof.
\end{proof}

}

\section{Introducing Linear and Affine Automata} \label{section:affine_introduction}

\emph{\chng{Linear} automata} are a much-studied class of CAs. Studying the automata simulated by them leads naturally to a broader class of \emph{affine automata}, which we now introduce. 

\begin{defn} [affine CA, \chng{linear} CA]
    Let $\F$ be a finite field, $V$ a finite-dimensional vector space over $\F$, and let $\A = (V^\Z, F)$ be a CA with local algebra $(V, f)_r$. We say that $\AA$ (or $\A$) is \emph{affine over $\F$} if $f: V^{2r+1} \rightarrow V$ is an affine mapping between vector spaces over $\F$. In such a case, we can write $f$ in the following form:
    \begin{align} \label{def:affine_f}
        f(\xx_{-r}, \ldots, \xx_r) = f_{-r}(\xx_{-r}) + \cdots + f_r(\xx_r) + \cc,
    \end{align}
    where $f_i: V \rightarrow V$ is a linear mapping for each $-r \leq i \leq r$ and $\cc \in V$ is a constant vector. The mapping $f_i$ is called the \emph{$i$-th component of $f$}. The class of all local algebras isomorphic to some affine local algebra over $\F$ with radius $r$ is denoted as $\AFF^\F_r$.

    In the special case when $f: V^{2r+1} \rightarrow V$ is a linear mapping between vector spaces over $\F$, we say that $\AA$ (or $\A$) is \emph{linear over $\F$}. In such a case, we can write $f$ as in (\ref{def:affine_f}) with $\cc = \vo$. We denote the class of all local algebras isomorphic to some local algebra linear over $\F$ with radius $r$ as \chng{$\LIN^\F_r$}. We say that $\AA$ is \emph{canonical linear (or canonical affine)} if it is linear (or affine) over $\F_p$ for some prime $p$ and $V = \F_p$. 
\end{defn}

\begin{ex}
    Consider the elementary CA 150 with local algebra $\ECA_{150} = (\F_2, f)_1$ where $f: \F_2^3 \rightarrow \F_2$ is defined as $f(x,y,z)=x+y+z \bmod 2$. Then, $\ECA_{150}$ is a CA linear over $\F_2$; in fact, it is canonical linear.
\end{ex}

Automata linear over finite fields of the form $\F_p$, $p$ prime, have been widely studied. In fact, they form one of the few classes of cellular automata that are amenable to algebraic analysis which yields rigorous results about their global dynamics. To name a few important results not directly related to CA simulation, the properties of linear CAs' global dynamics have been carefully analysed in the seminal work \cite{algebraic_properties_of_ca} and consequently in \chng{\cite{exact_results_for_determinstic_cas, linear_cas_and_recurring_sequences, self_similarity_of_linear_cas, linear_cas_finite_automata_pascals_triangle, gutschow2010fractal}}. \chng{We note that canonical linear CAs have also been called ``additive'' in some works, such as \cite{additive_cas_over_Zp_bottom}.}

The class of affine automata is a natural generalization of linear CAs that occurs when we study sub-automata of linear CAs. As shown further in the next section, it is the case that sub-automata of linear CAs need not be linear but are in general affine. To get acquainted with the two notions, we start with a simple lemma and an example.

\begin{lemma} \label{lemma:linear_idempotent}
    Let $\F$ be a finite field and $r \in \N$. Let $\AA = (V, f)_r$ be such that $\AA \in \AFF^\F_r$. Then, $\AA \in \LIN^\F_r$ if and only if there exists an idempotent element of $f$; i.e., there exists $\vv \in V$ such that $f(\vv, \ldots, \vv)=\vv$.
\end{lemma}
\begin{proof}
    Each linear CA has $\vo$ as an idempotent element. In the other direction, let $\AA = (V, f)_r$ be such that $\AA \in \AFF^\F_r$, and suppose that $\vv \in V$ is an idempotent of $f$. Let $f_{-r}, \ldots, f_r$ be the components of $f$. We define the bijection $\varphi: V \rightarrow V$ as $\varphi(\vv + \xx) = \xx$ for every $\xx \in V$. Next, we define $g:V^{2r+1} \rightarrow V$ as $g(\xx_{-r}, \ldots, \xx_r) \coloneqq \sum_{i=-r}^r f_i(\xx_i)$. Then, the CA with local algebra $(V, g)_r$ is linear, and it is straightforward to verify that $(V, f)_r \cong (V, g)_r$ via $\varphi$; indeed, for any $\xx_{-r}, \ldots, \xx_r$ we have:
    \begin{align*}
        \varphi\big(f(\xx_{-r}+\vv, \ldots, \xx_r+\vv)\big) &= \varphi\bigl(f(\vv, \ldots, \vv)+\sum_{i=-r}^r f_i(\xx_i)\bigr) = \sum_{i=-r}^r f_i(\xx_i)=\\ 
        &= \sum_{i=-r}^r f_i(\varphi\big(\xx_i+\vv)\big)=g\big(\varphi(\xx_{-r}+\vv), \ldots, \varphi(\xx_{r}+\vv)\big). \qedhere
    \end{align*} 
\end{proof}

\begin{ex}
    $\LIN^{\F_2}_1 \subsetneq \AFF^{\F_2}_1$. An example of an affine CA that is not isomorphic to any linear one is the elementary CA 105 with local algebra $(\F_2, f)_1$ where $f(x,y,z)=(x+y+z+1) \bmod 2$ for any $x,y,z \in \F_2$. It is clear that neither 0 nor 1 are idempotent elements of $f$.
\end{ex}

\chng{
Many of the results presented in this paper hold for affine CAs with at least two components that are bijections. Recall Definition \ref{def:bijectivity} which introduced the bijective condition and bi-permutivity. Below, we review the notions in the context of affine CAs.

\begin{obs}
Let $\AA=(V,f)_r$ be an affine local algebra of a CA with radius $r$ whose local rule $f$ has components $f_{-r}, \ldots, f_r$. $\AA$ satisfies the bijective condition if and only if at least two of the components are bijections. 
Let $-r \leq i < j \leq r$. $\AA$ is $(i,j)$-bi-permutive if and only if $f_i$ and $f_j$ are bijections, and for all $k< i$ and $k > j$, $f_k$ is the constant $\vo$ mapping.
\end{obs} 

When studying general CAs, bi-permutivity is a very restrictive condition \cite{topological_dynamics_of_cellular_automata}. As an example, it excludes the interesting class of reversible automata since bi-permutive CAs cannot be injective. However, the condition is not so confining for affine CAs. There, bi-permutivity means that the CA's ``outer components'' are represented by invertible matrices, and a random matrix is indeed much more likely invertible than singular. Moreover, Example \ref{ex:counterexample_bijective_condition} shows that bi-permutivity is an important restriction: the (weaker) bijective condition need not be preserved by the iterative powers.

\begin{defn}
Let $\AA=(V,f)_r$ be an affine local algebra of a CA with radius $r$ and let $-r \leq i < j \leq r$. We write that $\AA \in \AFF^\F_{r;i,j}$ if $\AA \in \AFF^\F_{r}$ and at the same time $\AA \in \BPR_{r;i,j}$. Analogously, we define the subclass $\LIN^\F_{r;i,j}$.
\end{defn}
}

Now we can state the main result of this paper that we prove in Section \ref{section:affine_limitations}:

\begin{thm}\label{thm:main_result}
    Let $p$ be a prime, $r \in \N$, and $-r \leq i < j \leq r$. Let $\AA$, $\BB$ be local algebras of cellular automata with radius $r$ such that $\BB \in \AFF^{\F_p}_{r;i,j}$. If $\AA \preceq \BB$, then $\AA \in \AFF^{\F_p}_{r;i,j}$.
\end{thm}

\subsection{Related Work on CA Simulation} \label{subsection:related_work}

In their seminal work \cite{inducing_order_on_ca_by_grouping}, Mazoyer and Rapaport study the properties of a strongly related notion of CA simulation based on iterative powers and sub-automata. In \cite{additive_cas_over_Zp_bottom}, they specifically focus on showing that a particular class of canonical linear CAs is limited in terms of what it can simulate. Subsequently, their work was continued in the collaboration \cite{bulking1}, where the authors introduce a generalized version of the iterative powers of CAs that we denote here by $\widetilde{\Xi}$ for clarity. Informally, $\widetilde{\Xi}$ allows for much more general geometrical transformations of the CA space-time diagrams. Whereas $\Xi$ is based on grouping together blocks of consecutive ``cells'' via the operators $o_n$, $\widetilde{\Xi}$ allows for groupings of much more general patterns, as well as, e.g., shifts of the CA's configuration space. In their subsequent work \cite{bulking2}, the same authors introduce various notions of CA simulation; the most general one is roughly defined below.

\begin{defn} [CA Simulation: Delorme, Mazoyer, Ollinger, Theyssier \cite{bulking2}]
    Let $\A$ and $\B$ be two 1D CAs both with radius $r \in \N$. For their local algebras, we write $\AA \preceq_\mathrm{m} \BB$ if $\widetilde{\Xi}(\AA) \cap \HHH\SSS\widetilde{\Xi}(\BB) \neq \emptyset$.
\end{defn}

Both lines of work (\cite{inducing_order_on_ca_by_grouping, additive_cas_over_Zp_bottom} and \cite{bulking1, bulking2}) analyse various important properties of the different CA simulations. This includes studying the simulation limitations of a special family of linear automata, namely those with local algebras of the form $\ZC_p = (\F_p, f)_1$ where $f(x, y, z)=(x+y+z) \bmod p$. Their negative result is that whenever $p \neq q$, one has $\ZC_p \npreceq_\mathrm{m} \ZC_q$. The authors use two observations:
\begin{enumerate}
    \item Let $\AA$, $\BB$ be local algebras of CAs with radius $r$. Then, $\widetilde{\Xi}(\AA) \cap \HHH\SSS\widetilde{\Xi}(\BB) \neq \emptyset$ if and only if $\widetilde{\Xi}(\AA) \cap \HHH\SSS\Xi(\BB) \neq \emptyset$.
    \item Let $\AA$ be a local algebra of a CA with $p$ states. Then the algebras in $\widetilde{\Xi}(\AA)$ have $p^k$ states, $k \geq1$.
\end{enumerate}
Given the observations, the authors reduce the problem of showing $\ZC_p \npreceq_\mathrm{m} \ZC_q$ for $p \neq q$ to proving that any algebra in $\HHH\SSS\Xi(\ZC_q)$ has $q^l$ elements for some $l \geq 0$. Assuming this extra piece of knowledge, which we deduce below as a corollary of our work, comparing the sizes of algebras in $\HHH\SSS\Xi(\ZC_q)$ and in $\widetilde{\Xi}(\ZC_p)$ already implies the negative result, as $p^k \ne q^l$ for $k\ge 1$, $l \geq 0$.

The crucial information about possible sizes of algebras in $\HHH\SSS\Xi(\ZC_q)$ follows immediately from our results. (This is not surprising, since Theorem \ref{thm:main_result} provides a lot of information about the structure of these algebras.) Indeed: For every prime $q$, clearly $\ZC_q \in \LIN^{\F_q}_{1;-1,1}$. By Theorem \ref{thm:main_result}, any local algebra in $\HHH\SSS\Xi(\ZC_q)$ belongs to $\AFF^{\F_q}_{1;-1,1}$; therefore it has an underlying structure of a vector space over $\F_q$ and hence, it has $q^l$ elements.

%

\section{Simulation Limitations of Linear and Affine Automata} \label{section:affine_limitations}

In this section, we prove Theorem \ref{thm:main_result}. In more detail, we will show that for any finite field $\F_p$, $p$ prime, any $r \in \N$, and any $-r \leq i < j \leq r$, it holds that:
\begin{align}
    \HHH\SSS\Pfin \Xi(\AFF^{\F_p}_{r;i,j})&=\AFF^{\F_p}_{r;i,j} \label{res:aff},\\
    \HHH\SSS\Pfin \Xi(\LIN^{\F_p}_{r;i,j})&=\AFF^{\F_p}_{r;i,j} \label{res:add}.
\end{align}
Theorem \ref{thm:main_result} is then a direct consequence of (\ref{res:aff}). Concretely, we will show this result by studying how the operators $\Xi$, $\Pfin$, $\SSS$, and $\HHH$ change the sets $\AFF^{\F_p}_{r;i,j}$ and $\LIN^{\F_p}_{r;i,j}$.

\begin{lemma}\label{lemma:ksi_affinity}
    Let $\AA$ be a local algebra of a CA that is affine (or linear) over a finite field $\F$ and let $n \in \N$. Then $\AA^{[n]}$ is again affine (or linear) over $\F$.  
\end{lemma}
\begin{proof}
    Let $\AA=(V, f)_r$ where $f$ has linear components $f_{-r}, \ldots, f_r$ and $f(\xx_{-r}, \ldots, \xx_r) = f_{-r}(\xx_{-r}) + \cdots + f_r(\xx_r) + \cc$ for some constant $\cc \in V$; $\xx_{-r}, \ldots, \xx_r \in V$. We put $g(\xx_{-r}, \ldots, \xx_r) \coloneqq f(\xx_{-r}, \ldots, \xx_r) - \cc$ and $\BB \coloneqq (V, g)_r$. Clearly, $\BB$ is linear over $\F$. Let $n \in \N$; we first show that $\BB^{[n]}$ is also linear. It is straightforward to verify that for each $k \geq 2r+1$, $\restr{\widetilde{g}}{V^k}: V^k \rightarrow V^{k-2r}$ is a linear mapping of vector spaces. Thus, since composition of linear mappings is linear, we have that $\restr{\widetilde{g}^n}{V^{n(2r+1)}}: V^{n(2r+1)} \rightarrow V^n$ is also linear. It is then a straightforward yet slightly technical step to verify that $g^{[n]}$ is a linear mapping and thus, that $\BB^{[n]}$ is linear over $\F$.

    Next, we observe that $\widetilde{f}(\xx_1 \cdots \xx_{k})=\widetilde{g}(\xx_1 \cdots \xx_{k})+\underbrace{\cc \cc \cdots \cc\cc}_{(k-2r) \times}$ and that $\widetilde{f}\big((\xx_1 \cdots \xx_{k}) + (\yy_1 \cdots \yy_{k})\big)=\widetilde{g}(\xx_1 \cdots \xx_{k})+\widetilde{f}(\yy_1 \cdots \yy_{k})$ for any $k \geq 2r+1$ and any $\xx_1, \ldots, \xx_k, \yy_1, \ldots, \yy_k \in V$.
    Combining these two facts, it is straightforward to verify by induction on $n$ that for any $k \geq 2nr+1$, we have 
    $$\widetilde{f}^n(\xx_1 \cdots \xx_{k}) = \widetilde{g}^n(\xx_1 \cdots \xx_{k}) + \widetilde{f}^{n-1}(\underbrace{\cc \cc \cdots \cc\cc}_{(k-2r) \times})$$
    for each $\xx_1, \ldots, \xx_{k} \in V$. For $k=n(2r+1)$ we get that $f^{[n]}(\cdot) = g^{[n]}(\cdot) + \dd$ for a constant $\dd \in V^n$ and thus, that $\AA^{[n]}$ is affine over $\F$.
\end{proof}

\begin{cor} \label{cor:xi(aff)=aff} Let $\F$ be an arbitrary finite field, $r \in \N$, and $-r \leq i < j \leq r$. Then:
\begin{equation*}
\begin{aligned}[c]
\Xi(\AFF_r^\F)&=\AFF_r^\F,\\
\Xi(\AFF_{r;i,j}^\F)&=\AFF_{r;i,j}^\F,
\end{aligned}
\qquad \qquad
\begin{aligned}[c]
\Xi(\LIN_r^\F)&=\LIN_r^\F,\\
\Xi(\LIN_{r;i,j}^\F)&=\LIN_{r;i,j}^\F.
\end{aligned}
\end{equation*}
\end{cor}
\chng{
\begin{proof}
    This is a direct consequence of Lemma \ref{lemma:ksi_affinity} together with Lemma \ref{lemma:ksi_permutivity}.
\end{proof}
}

\begin{obs}
    Let $\AA_1, \ldots, \AA_k$ be local algebras of CAs with radius $r$ that are affine (or linear) over a finite field $\F$. Then, $\AA = \AA_1 \times \cdots \times \AA_k$ is again affine (or linear) over $\F$. Moreover, let $-r \leq i < j \leq r$. If the algebras $\AA_1, \ldots, \AA_k$ all have their $i$-th and $j$-th components bijective, then so does $\AA$.
\end{obs}

\begin{cor} \label{cor:pfin(aff)=aff} Let $\F$ be an arbitrary finite field, $r \in \N$, and $-r \leq i < j \leq r$. It holds that:
\begin{equation*}
\begin{aligned}[c]
\Pfin(\AFF_r^\F)&=\AFF_r^\F,\\
\Pfin(\AFF_{r;i,j}^\F)&=\AFF_{r;i,j}^\F,
\end{aligned}
\qquad \qquad
\begin{aligned}[c]
\Pfin(\LIN_r^\F)&=\LIN_r^\F,\\
\Pfin(\LIN_{r;i,j}^\F)&=\LIN_{r;i,j}^\F.
\end{aligned}
\end{equation*}
\end{cor}

Thus, we have shown that both the operators $\Xi$ and $\Pfin$ preserve the class of $\AFF^\F_{r;i,j}$ for each finite field $\F$, each radius $r \in \N$, and each $-r \leq i < j \leq r$. Below, we show similar results for the operators $\SSS$ and $\HHH$ under the assumption that the affine automata satisfy the bijective condition. In what follows, $p$ always denotes a prime number.

\subsection{Sub-automata of affine CAs} \label{subsection:subautomata_of_affine_cas}
In this section, we study the operator $\SSS$ on the class of affine automata. Compared to the part about $\Pfin$ and $\Xi$, our results need more assumptions: We only work over the prime fields $\F_p$ and we require the bijective condition (see Definition \ref{def:bijectivity}). The importance of the assumptions is illustrated in Example \ref{ex:counterexample_s} where we exhibit an affine CA violating the bijective condition that contains a sub-automaton that is not affine.

We start by noticing that certain invariant subspaces produce a natural family of sub-automata.

\begin{obs}
    Let $\BB=(V, f)_r$ be the local algebra of a CA affine over a \chng{finite field $\F$}. Suppose that $W \leq V$ is a subspace invariant under all components of $f$ and that $\vv \in V$ satisfies $f(\vv, \ldots, \vv) \in \vv +W$. Then $\AA = (\vv+W, \restr{f}{(\vv+W)^{2r+1}})_r$ belongs to $\SSS(\BB)$. 
\end{obs}

Assuming the bijective condition, the observation can be turned into an equivalence -- every sub-automaton is of the simple form described above.

\begin{prop} \label{lemma:affine_subautomata}
    Let $\B$ be a CA with local algebra $\BB=(V,f)_r$ that is affine over $\F_p$ and satisfies the bijective condition. Let $\A$ be a sub-automaton of $\B$ with the local algebra $\AA=(U, \restr{f}{U^{2r+1}})_r$. Then, $U = \vv + W$ for a subspace $W \leq V$ invariant under all components of $f$ and $\vv\in\V$ such that $f(\vv, \ldots, \vv) \in \vv + W$.
\end{prop}
\begin{proof}  
    Let $f_{-r}, \ldots, f_r$ be the components of $f$ and $f(\xx_{-r}, \ldots, \xx_r) = \sum_{i=-r}^r f_i(\xx_i) + \cc$ for some $\cc \in V$.
    Let us fix $-r \leq i < j \leq r$ such that $f_i$ and $f_j$ are bijections. We put $U_i \coloneqq f_i(U) \subseteq V$ and $U_j \coloneqq f_j(U) \subseteq V$ and we note that $|U_i|=|U_j|=|U|\eqqcolon k.$ We first show that both $U_i$ and $U_j$ are affine subspaces of $V$. We have that: 
    $$k = |U| \geq |f(U, \ldots, U)| = |f_{-r}(U) + \cdots + f_r(U)| \geq |f_i(U)+f_j(U)| = |U_i + U_j| \geq |U_i| = k.$$ 
    Thus, we also have $|U_i+U_j|=k$. Let $\uu_i \in U_i$ and $\uu_j \in U_j$. 
    We put $W_i \coloneqq U_i-\uu_i$ and $W_j \coloneqq U_j-\uu_j$. 
    Then clearly $\vo \in W_i$ and $\vo \in W_j$ and moreover, since translations do not change the set sizes, $|W_i|=|W_j|=|W_i+W_j|=k$. 
    Hence, $W_i+ W_j$ is equal to some set $W$ and $W \supseteq W_i + \vo = W_i$. 
    Since $W_i$ has the same size as $W$, this yields $W_i = W$. Analogously, we have $W_j = W$ and thus $W = W_i = W_j = W_i + W_j = W+W$. This proves that $W$ is closed under addition, which is over $\F_p$ sufficient to make it a vector subspace of $V$. Finally, $U_i = \uu_i+ W$ and $U_j = \uu_j + W$ are both affine subspaces of $V$, one being just a translation of the other.
    
    Now, pick any $\uu \in U$ and put $\vv \coloneqq \left( \sum_{l=-r}^r f_l(\uu) \right) - f_i(\uu) + \uu_i+\cc$. Then:
    $$U \supseteq f(U, \ldots, U) \supseteq f_{-r}(\uu) + \cdots + f_{i-1}(\uu) + U_i + f_{i+1}(\uu) + \cdots + f_{r}(\uu) + \cc = \vv + W.$$ 
    Since $|U|=|W|$, we get $U = \vv + W$.

    It is left to show that the subspace $W$ is invariant under all components of $f$. For any fixed $k$ with $-r \leq k \leq r$, we show that $f_k(W) \leq W$:
    $$\vv + W \supseteq f(\vv+W, \ldots, \vv+W) = f(\vv, \ldots, \vv) + f_{-r}(W) + \cdots + f_r(W) \supseteq f(\vv, \ldots, \vv) + f_k(W).$$
    Since $f_k(W)$ is a subspace, this immediately yields $f(\vv, \ldots, \vv) \in \vv + W$ as well as $f_k(W) \le W$, which concludes the proof.  
\end{proof}

At this point, it remains to show that the sub-automaton $\AA = (\vv+W, \restr{f}{(\vv+W)^{2r+1}})_r$ from the previous proposition is indeed isomorphic to an affine automaton (whose states must form a vector space, not an affine subspace).
\begin{cor} \label{cor:s(aff)=aff} Let $r \in \N$, $-r \leq i < j \leq r$. It holds that $\SSS(\AFF^{\F_p}_{r;i,j}) = \AFF^{\F_p}_{r;i,j}$.
\end{cor}
\begin{proof}
    Take a $\BB = (V, f)_r \in \AFF^{\F_p}_{r;i,j}$, and consider any $\AA \in \SSS(\BB)$. From Proposition \ref{lemma:affine_subautomata}, we already know that $\AA$ is of the form $\AA = (\vv+W, \restr{f}{(\vv+W)^{2r+1}})_r$ where $W \leq V$ is a subspace invariant under all components of $f$, $\vv \in V$, and $f(\vv, \ldots, \vv) \in \vv+W$. We need to construct a local algebra $\AA'$ affine over $\F_p$ such that $\AA \cong \AA'$.

    There is a $\ww \in W$ such that $f(\vv, \ldots, \vv) = \vv + \ww$. We define $\AA' = (W, g)_r$ where for any $\xx_{-r}, \ldots, \xx_r \in W$, $g(\xx_{-r}, \ldots, \xx_r) \coloneqq \sum_{i=-r}^r f_i(\xx_i) +\ww$. Since $W$ is invariant under components of $f$, indeed $g: W^{2r+1} \rightarrow W$. Thus, $(W, g)_r$ is affine over $\F_p$ with linear components $g_k = \restr{f_k}{W}$; in particular, $g_i$ and $g_j$ are bijections. Let us define $\varphi: \vv + W \rightarrow W$ by $\varphi(\vv + \xx) = \xx$ for any $\xx \in W$. Then, it is straightforward to verify that $\AA \cong \AA'$ via $\varphi$. Indeed, for any $\xx_{-r}, \ldots, \xx_r \in W$, we have:
    \begin{align*}    
        \varphi\big(f(\vv+\xx_{-r}, \ldots, \vv+\xx_r)\big)&=\varphi \big(f(\vv, \ldots, \vv)+\sum_{i=-r}^r f_i(\xx_i)\big)\\
        &= \varphi\big(\vv + \ww +\sum_{i=-r}^r f_i(\xx_i)\big) = \ww+ \sum_{i=-r}^r f_i(\xx_i) = g(\xx_{-r}, \ldots, \xx_r). \qedhere 
    \end{align*}
\end{proof}

\subsection{Sub-automata of linear CAs}
This section is not part of the proof of Theorem \ref{thm:main_result}; rather, it shows that the analogy for linear automata does not hold, as a sub-automaton of a linear automaton need not be linear. Rather, it turns out that affine automata can be introduced as sub-automata of linear ones.
\begin{prop} \label{prop:s(add)=aff}
    Let $r \in \N$, $-r \leq i < j \leq r$. It holds that $\SSS(\LIN^{\F_p}_{r;i,j}) = \AFF^{\F_p}_{r;i,j}$.
\end{prop}
\begin{proof}
     Corollary \ref{cor:s(aff)=aff}, together with the fact that linear automata are a subclass of affine ones, yields $\SSS(\LIN^{\F_p}_{r;i,j}) \subseteq \SSS(\AFF^{\F_p}_{r;i,j}) = \AFF^{\F_p}_{r;i,j}$. We complement the results by showing that $\AFF^{\F_p}_{r;i,j} \subseteq \SSS(\LIN^{\F_p}_{r;i,j})$. Let $\A$ be an affine CA with local algebra $\AA = (V, f)_r$, $\AA \in \AFF^{\F_p}_{r;i,j}$. We construct a CA $\B$ with local algebra $\BB \in \LIN^{\F_p}_{r;i,j}$ such that $\AA \in \SSS(\BB)$. As always, write $f:V^{2r+1} \rightarrow V$ as $f(\xx_{-r}, \ldots, \xx_r) = f_{-r}(\xx_{-r}) + \cdots + f_r(\xx_r) + \cc$ where $f_i: V \rightarrow V$ are linear mappings and $\cc \in V$ is a constant. We put $W = V \times \F_p$. We fix a basis $(\vv_1, \ldots, \vv_{k-1})$ of $V$; then, $B \coloneqq (\ww_1, \ldots, \ww_{k-1}, \ww_{k}) \coloneqq \left((\vv_1, 0), \ldots, (\vv_{k-1},0), (\vo, 1)\right)$ is a basis of $W$. We define linear mappings $g_i: W \rightarrow W$ for each $-r \leq i \leq r$ on the basis $B$ as follows:
    \begin{equation*}
    \begin{aligned}[c]
    g_i(\ww_j) &\coloneqq (f_i(\vv_j), 0)\\
    g_i(\ww_k) &\coloneqq
         \begin{cases*}
      (\vo, 1)\\
      (\vo, -1)\\
      (-\cc, 1)
      \end{cases*}
    \end{aligned}
    \qquad \qquad
    \begin{aligned}[c]
    &\text{ for all } -r \leq i \leq r \text{ and } 1 \leq j \leq k-1,\\
    &\text{ for } -r \leq i < 0,\\
    &\text{ for } 0 < i \leq r,\\
    &\text{ for } i=0.
    \end{aligned}
    \end{equation*}
     
     We put $g(\xx_{-r}, \ldots, \xx_{r}) = g_{-r}(\xx_{-r}) + \cdots + g_{r}(\xx_r) + (\cc, 0)$ for any $\xx_{-r}, \ldots, \xx_r \in W$ and define the automaton $\B$ with local algebra $\BB = (W, g)_r$. Clearly $\AA \cong (V\times\{0\}, \restr{g}{(V\times\{0\})^{2r+1}})$, so indeed $\AA \in \SSS(\BB)$. From the construction, it is clear that for each $-r \leq i \leq r$, $g_i$ is bijective whenever $f_i$ is. Furthermore, $g(\ww_k, \ldots, \ww_k) = \underbrace{(\vo, 1) + \cdots + (\vo, 1)}_{r \, \times} + (- \cc, 1) + \underbrace{(\vo, -1) + \cdots + (\vo, -1)}_{r \, \times} + (\cc, 0) = (\vo, 1)=\ww_k$. Thus, $g$ has an idempotent element, and due to Lemma \ref{lemma:linear_idempotent}, $\BB \in \LIN^{\F_p}_{r;i,j}$.
\end{proof}

\chng{
In fact, we have shown $\AFF^{\F}_{r;i,j} \subseteq \SSS(\LIN^{\F}_{r;i,j})$ for every finite field $\F$, since this part of the proof does not use any special property of $\F_p$.
}

\subsection{Quotient automata of affine CAs}
In this section, we study the operator $\HHH$ on the class of affine automata. Again, the bijective condition is required; in Example \ref{ex:counterexample_h} we construct an affine CA violating the bijective condition that contains a quotient automaton that is not affine.

We start by a simple observation: Every invariant subspace gives rise to a congruence and thus, to a quotient automaton. Subsequently, we show the converse: If an affine CA $\BB$ satisfies the bijective condition, then each congruence on $\BB$ is already a congruence on the underlying vector space. To complement the result, in Example \ref{ex:counterexample_h} we construct an affine CA violating the bijective condition with a congruence that is not a congruence on the vector space.

\begin{obs} \label{obs:affine_quotient_automata}
    Let $\B$ be an affine CA over a finite field $\F$ with local algebra $\BB = (V, f)_r$. Let $W \leq V$ be a subspace invariant under all components of $f$. We define $\sim \, \subseteq \, V \times V$ as follows: $\uu \sim \vv$ if and only if $\uu - \vv \in W$. Then, $\sim$ is a congruence on $\BB$.
\end{obs}

In several steps, we now proceed to show that under the above-mentioned conditions, every quotient automaton is of the form described by Observation \ref{obs:affine_quotient_automata}.

\begin{lemma} \label{lemma:affine_quotient_automata}
    Let $\B$ be a CA with local algebra $\BB = (V, f)_r$ that is affine over a finite field $\F$ and satisfies the bijective condition. Denote the linear components of $f$ by $f_k: V \rightarrow V$; $-r \leq k \leq r$.
    For every congruence $\sim \, \subseteq V \times V$ on $\BB$ and every $\uu, \uu', \vv, \vv' \in V$ it holds that:
    \begin{enumerate}
        \item \label{lq1} If $\uu \sim \uu'$ then $f_k(\uu) \sim f_k(\uu')$ for all $-r \leq k \leq r$.
        \item \label{lq2} Moreover, if $f_k$ is a bijection for some $-r \leq k \leq r$, then $f_k(\uu) \sim f_k(\uu')$ implies $\uu \sim \uu'$.
        \item \label{lq3} If $\uu \sim \uu'$ and $\vv \sim \vv'$ then $\uu + \vv \sim \uu' + \vv'$.
        \item \label{lq4} If $\F=\F_p$ for some prime $p$, then $[\vo]_\sim = \{\xx \in V \mid \xx \sim \vo \}$ is a subspace of $V$ invariant under all components of $f$.
    \end{enumerate}
\end{lemma}
\begin{proof}
Throughout the proof, we write $f(\xx_{-r}, \ldots, \xx_r) = \sum_{i=-r}^r f_i(\xx_i) + \cc$ with components $f_k: V \rightarrow V$, $-r \leq k \leq r$ and with $\cc \in V$. We denote the two bijective components by $f_i$ and $f_j$; note that this is more general than assuming $\BB \in \AFF_{r;i,j}^{\F}$, as we do not require $f_i$ and $f_j$ to be the ``outer'' components.

    \emph{\ref{lq1}.}: Let $-r \leq k \leq r$ such that $k \neq i$. Since $f_i$ is a bijection, there exists $\bb \in V$ such that $f_i(\bb)=-\cc$. Then:
\begin{align*}
     &f_k(\uu) = f_k(\uu) -\cc +\cc 
                            = f(\vo, \ldots, \vo, \uu, \vo, \ldots, \vo, \bb, \vo, \ldots, \vo) \sim\\
    &\sim f(\vo, \ldots, \vo, \uu', \vo, \ldots, \vo, \bb, \vo, \ldots, \vo) 
                            = f_k(\uu') -\cc +\cc
                            = f_k(\uu'), 
\end{align*}
where $\uu, \uu'$ are on the $k$-th position and $\bb$ is on the $i$-th. The proof is analogous if $k=i$ as we can use the bijective component $f_j$.

    \emph{\ref{lq2}.}: If $f_k$ is bijective and $f_k(\uu) \sim f_k(\uu')$, we can repeatedly apply part \ref{lq1}.\ to get $f_k^n(\uu) \sim f_k^n(\uu')$ for every $n \in \N$. Since $f_k: V \rightarrow V$ is an automorphism of a finite-dimensional vector space over a finite field, there exists $n$ such that $f_k^n = \id$. Thus, $\uu = f_k^n(\uu) \sim f_k^n(\uu') = \uu'$.

    \emph{\ref{lq3}.}: Let $\uu \sim \uu'$. We first show that $\uu -\cc \sim \uu' - \cc$. There exist some $\aa, \aa' \in V$ such that $f_i(\aa)=\uu$ and $f_i(\aa')=\uu'$ and from \emph{\ref{lq2}.}, it holds that $\aa \sim \aa'$. There also exists a $\bb \in V$ such that $f_j(\bb) = -\cc-\cc$. Then: 
    \begin{align*}
        &\uu - \cc = f_i(\aa) + f_j(\bb) + \cc = f(\vo, \ldots, \vo, \aa, \vo, \ldots, \vo, \bb, \vo, \ldots, \vo) \sim\\
        &\sim f(\vo, \ldots, \vo, \aa', \vo, \ldots, \vo, \bb, \vo, \ldots, \vo) = f_i(\aa') + f_j(\bb) + \cc = \uu' - \cc,
    \end{align*}
    where $\aa, \aa'$ are on the $i$-th position and $\bb$ is on the $j$-th. Next, let $\vv \sim \vv'$ and let $\yy, \yy' \in V$ be such that $f_j(\yy)=\vv$ and $f_j(\yy')=\vv'$. Since $\uu - \cc \sim \uu' -\cc$, we find $\xx, \xx' \in V$ such that $f_i(\xx)=\uu-\cc$, $f_i(\xx')=\uu'-\cc$, and from \emph{\ref{lq2}.}, we again have $\xx \sim \xx'$. Then:
    \begin{align*}
        &\uu + \vv = (\uu-\cc) + \vv + \cc = f_i(\xx) + f_j(\yy) + \cc =  f(\vo, \ldots, \vo, \xx, \vo, \ldots, \vo, \yy, \vo, \ldots, \vo) \sim\\
        &\sim f(\vo, \ldots, \vo, \xx', \vo, \ldots, \vo, \yy', \vo, \ldots, \vo) = f_i(\xx') + f_j(\yy') +\cc = (\uu' -\cc) +\vv' +\cc = \uu' + \vv',
    \end{align*}
    where $\xx, \xx'$ are on the $i$-th position and $\yy, \yy'$ are on the $j$-th.

    \emph{\ref{lq4}.}: By \emph{\ref{lq3}.}, the congruence $\sim$ ``respects addition'' in $V$. If $\F=\F_p$, then multiplication of vectors from $V$ by scalars from $\F$ is generated by addition in $V$, and therefore $\sim$ also ``respects multiplication by scalars''. Hence, $\sim$ is not only a congruence on the algebra but also a congruence on the vector space $V$, and thus $[\vo]_\sim \leq V$. By \emph{\ref{lq1}.}, the subspace $[\vo]_\sim$ is invariant under all components of $f$. 
\end{proof}

We now need to exploit the information about $[\vo]_\sim$ to infer knowledge about the quotient space.

\begin{cor}  \label{cor:h(aff)=aff}       
        Let $r \in \N$ and $-r \leq i < j \leq r$. Then $\HHH(\AFF^{\F_p}_{r;i,j}) = \AFF^{\F_p}_{r;i,j}$. Moreover, $\HHH(\LIN^{\F_p}_{r;i,j}) = \LIN^{\F_p}_{r;i,j}$.
\end{cor}
\begin{proof}
    Let $\B$ be an affine CA with local algebra $\BB = (V, f)_r \in \AFF^{\F_p}_{r;i,j}$, and consider a congruence $\sim$ on $\BB$. Let $\AA = \bigslant{\BB}{\sim} = (\bigslant{V}{\sim}, h)_r$ be the quotient algebra of $\BB$. By Lemma \ref{lemma:affine_quotient_automata}, part \emph{\ref{lq4}.}, $\sim$ is a congruence on the vector space $V$ and thus, $\bigslant{V}{\sim}$ is again a vector space over $\F_p$. Let $W = [\vo]_\sim$. Then $\bigslant{V}{\sim} = \{\xx + W \mid \xx \in V\}$. For each $-r \leq l \leq r$ we define $h_l(\xx+W) \coloneqq f_l(\xx)+W$ for any $\xx \in V$. Thanks to Lemma \ref{lemma:affine_quotient_automata}, \emph{\ref{lq1}.}, $h_l: \bigslant{V}{\sim} \rightarrow \bigslant{V}{\sim}$ is a well-defined mapping. Clearly, it is a linear mapping on $\bigslant{V}{\sim}$. By definition, $h(\xx_{-r}+W, \ldots, \xx_r+W) = f(\xx_{-r}, \ldots, \xx_r) + W = \sum_{l=-r}^r f_l(\xx_l) + \cc + W$. This is further equal to $\sum_{l=-r}^r (f_l(\xx_l)+W) + (\cc + W) = \sum_{l=-r}^r h_l(\xx_l + W) + (\cc + W)$. This shows, as expected, that $h$ is indeed an affine mapping with components $h_l$, $-r \leq l \leq r$. Clearly, if $f_l$ is a bijection for some $-r \leq l \leq r$, then $h_l$ is surjective and therefore also a bijection. Thus, $\AA = \bigslant{\BB}{\sim} = (\bigslant{V}{\sim}, h)_r$ is a CA affine over $\F_p$ whose local rule has its $i$-th and $j$-th component bijective. This finishes the proof in the affine case.

    To conclude the second part of the statement: If $\BB \in \LIN^{\F_p}_{r;i,j}$, then $\vo \in V$ is an idempotent of $f$, so it is easy to see that $[\vo]_\sim$ is an idempotent for $\AA = \bigslant{\BB}{\sim}$. Hence, \chng{by Lemma \ref{lemma:linear_idempotent},} $\AA \in \LIN^{\F_p}_{r;i,j}$.
\end{proof}

\subsection{Main Result and Examples}

Combining Corollaries \ref{cor:xi(aff)=aff}, \ref{cor:pfin(aff)=aff}, \ref{cor:s(aff)=aff}, \ref{cor:h(aff)=aff}, and Proposition \ref{prop:s(add)=aff} yields the following main result. Note that the announced Theorem \ref{thm:main_result} is just another formulation of (\ref{thm:res1}).

\begin{thm} \label{thm:full_main_result}
    For any $r \in \N$, $-r \leq i < j \leq r$, and any prime $p$ it holds that  
    \begin{align}
    \HHH\SSS\Pfin \Xi(\AFF^{\F_p}_{r;i,j})&=\AFF^{\F_p}_{r;i,j}, \label{thm:res1}\\
    \HHH\SSS\Pfin \Xi(\LIN^{\F_p}_{r;i,j})&=\AFF^{\F_p}_{r;i,j}. \label{thm:res2}
\end{align}
\end{thm}

\chng{
In what follows, we provide a few examples of affine automata not satisfying the assumptions of Theorem \ref{thm:full_main_result} for which the result does not hold. While we defined linear and affine automata over general finite fields, the first example illustrates that our results only work over the fields $\F_p$, $p$ prime. Over them, the multiplication can be recovered by repeated addition, which is why such linear CAs are sometimes also called \emph{additive}.

\begin{ex}[non-affine sub-automaton when the field is not prime] \label{ex:counterexample_finite_field}
Let $\F_{p^n}$ be the field with $p^n$ elements, $n>1$. Consider the local algebra $\BB=(\F_{p^n}, f)_r$ of a CA linear over $\F_{p^n}$ where
$$
f(\xx_{-r}, \ldots, \xx_r) = a_{-r}\xx_{-r} + \cdots + a_r\xx_r
$$
with the constants $a_i$ belonging to the prime subfield $\F_p$. Assume further that at least two of them are non-zero.

Since $\F_{p^n}$ is an $\F_p$-vector space, we know that $\BB$ is affine over $\F_p$ as well. Thus, we can use Proposition \ref{lemma:affine_subautomata} to get a full characterisation of sub-automata of $\BB$: They correspond to (shifted) $\F_p$-subspaces of the $n$-dimensional space $\F_{p^n}$. Specifically, the size of the sub-automata can only be of the form $p^k$, $0 \leq k \leq n$. Hence, the only two sub-automata of $\BB$ that are affine over $\F_{p^n}$ are the trivial ones: $\BB$ and the automaton with a single element. On the other hand, all of them are $\F_p$-affine.
\end{ex}

When analysing the sub-automata and quotient automata of affine CAs, we further relied on the assumption that the CA has two of its components bijective. We now give simple counterexamples of affine CAs violating the bijective condition that contain a sub-automaton or a quotient automaton that is not affine. 

While in Example \ref{ex:counterexample_finite_field}, the argument relied simply on the cardinality of the sub-automata, the following examples are more elaborate. Namely, we show the existence of a CA affine over $\F_p$ violating the bijective condition whose sub-automaton (Example \ref{ex:counterexample_s}) or quotient automaton (Example \ref{ex:counterexample_h}) is not affine over any finite field. In fact, the examples already imply a much stronger result: the state space of the sub-automaton and quotient automaton cannot even be equipped with the structure of an abelian group to create an \emph{affine abelian CA}, discussed in detail in Subsection \ref{subsection:generalizations}.

\begin{ex}[non-affine sub-automaton when bijective condition is violated] \label{ex:counterexample_s}
    Let $V$ be a $d$-dimensional vector space over $\F_p$ and $\BB=(V, f)_r$ be a local algebra of a CA linear over $\F_p$ where $f$ has components $f_i:V \rightarrow V$, $-r \leq i \leq r$. Suppose that $f_0$ is bijective and that there exists a subset $S \subseteq V$ invariant under $f_0$ such that $f_k(S)=\{\vo \}$ for all $k \neq 0$. Let $\AA = (S, \restr{f}{S^{2r+1}})_r$. Clearly, $\AA$ is a sub-automaton of $\BB$ which is insensitive to all coordinates except for $0$.
    
    If $|S|$ is divisible by two different primes, we can quickly conclude that $\AA$ is not affine over any finite field. This is already a simple counterexample which illustrates the importance of the bijective condition. Below, we show that even if $S$ has an ``admissible'' cardinality, $\AA$ may not be affine over any finite field. Suppose that:
    \begin{itemize}
        \item There exist $\ww_1, \ww_2 \in S$, $\ww_1 \neq \ww_2$ such that $f_0(\ww_1)=\ww_1$ and $f_0(\ww_2)=\ww_2$.
        \item There exists $\ww' \in S$ such that $f_0(\ww')\neq\ww'$.
        \item $|S|=q$ for some prime $q$.
    \end{itemize}
    Clearly, if $\AA$ was affine, it would have to hold that $\AA \in \AFF_r^{\F_{q}}$. We will show this is not true. For contradiction, let us assume that there exists an isomorphism $\varphi: \AA \to (\F_q, g)_r$ where $g(\xx_{-r}, \ldots, \xx_r)=a\xx_0 + b$ for some $a, b \in \F_q$. Since $\ww_1$ and $\ww_2$ are fixed points of $f_0$, it holds that $a\varphi(\ww_1)+b=\varphi(\ww_1)$ and $a\varphi(\ww_2)+b=\varphi(\ww_2)$. This already implies that $a=1$ and $b=0$. However, we also have $a\varphi(\ww')+b\neq\varphi(\ww')$ which is a contradiction.
\end{ex}

\begin{ex}[non-affine quotient automaton when bijective condition is violated] \label{ex:counterexample_h}
    Let $V$ be a $d$-dimensional vector space over $\F_p$, $d \geq 3$, and $\BB=(V, f)_r$ be a local algebra of a CA linear over $\F_p$ where $f$ has components $f_i:V \rightarrow V$, $-r \leq i \leq r$. Suppose $f_0$ is bijective and further assume that:
    \begin{itemize}
        \item $f_k$ is the $\vo$ mapping for all $k \neq 0$.
        \item There exist $\uu, \vv \in V$, $\uu \neq \vv$, such that $f_0(\uu)=\vv$ and $f_0(\vv)=\uu$.
    \end{itemize}
Note that the second assumption already implies $f_0(\uu+\vv)=\uu+\vv$. And further, since $d \geq 3$, that there exists $\uu' \in V \setminus \{\uu, \vv \}$ which is not a fixed point of $f_0$.   
We define $\sim \, \coloneqq \{(\xx, \xx) \mid \xx \in V \} \cup \{(\uu, \vv),  (\vv, \uu) \}$. It is easy to verify that $\sim$ is a congruence on $\BB$ but is not a congruence on the vector space $V$. Let $\AA \coloneqq \bigslant{\BB}{\sim}$. Again, since $|\AA|=p^d-1$ cannot be a power of $p$, we can quickly conclude that $\AA$ is not affine over $\F_p$. Further, if $|\AA|$ is divisible by two different primes, we can again quickly conclude that $\AA$ is not affine over any finite field. In what follows, we show a stronger result: even when $|\AA|=q$ for some prime $q$ (which actually means that $p=2$ and $q$ is a Mersenne prime), we prove that $\AA$ is not affine over $\F_q$.

For contradiction, let us assume that $\varphi: \AA \to (\F_q, g)_r$ is an isomorphism where $g(\xx_{-r}, \ldots, \xx_r)=a\xx_0 + b$ for some $a, b \in \F_q$. Clearly, $a\varphi([\uu]_\sim)+b=[\vv]_\sim=[\uu]_\sim$ and $a\varphi([\vo]_\sim)+b=[\vo]_\sim$. This already implies $a=1$ and $b=0$. However, we also have that $a\varphi([\uu']_\sim)+b\neq[\uu']_\sim$ which yields a contradiction.
\end{ex}
}

\chng{
The last example illustrates an application of our theory: For a well-known elementary CA, we characterize all elementary CAs simulated by it.

\begin{ex}[ECA 60] \label{eca_60}
    Let us consider the class of elementary CAs; i.e., CAs with states $\F_2$ and radius $r=1$. We consider elementary CA 60 that is defined as $\ECA_{60}=(\F_2, f)_1$ where $f(x,y,z)=x+y \bmod 2$ and ECA 195 defined as $\ECA_{195}=(\F_2, g)_1$ where $g(x,y,z)=x+y+1 \bmod 2$. Clearly, $\ECA_{195} \cong \ECA_{60}$ via the bijection that exchanges 0 and 1. Further, $\ECA_{60}, \ECA_{195} \in \AFF^{\F_2}_{1; -1,0}$ and they are the only two elementary CAs that belong to this class. Thus, Theorem \ref{thm:full_main_result} implies that the only elementary CA that can be simulated by ECA 60 is itself (up to isomorphism).
\end{ex}
}

\paragraph{Canonical Affine Cellular Automata} Canonical affine CAs form the most studied subclass of affine automata. Our results can be used to analyse $\HHH\SSS \Pfin \Xi (\AA)$ in all non-trivial cases, as explained below. Suppose $\A$ is affine over $\F_p$ with local algebra $\AA = (\F_p, f)_r$ and has its local rule of the form $f(x_{-r}, \ldots, x_r) = a_{-r} x_{-r} + \cdots + a_r x_{r} + c$ for some coefficients $a_{-r}, \ldots, a_r \in \F_p$ and $c \in \F_p$. Then, we can distinguish the following cases:
\begin{enumerate}
    \item \label{can_add1} $a_i=0$ for all $-r \leq i \leq r$. In such a case, the CA is a constant zero mapping and studying $\HHH\SSS\Pfin \Xi (\AA)$ is trivial. Concretely, $$\HHH\SSS\Pfin \Xi (\AA) = \{(S,g)_r \mid S \text{ an arbitrary finite set and } g \text{ a constant mapping} \}.$$
    \item  \label{can_add2} $a_i\neq0$ for some $-r \leq i \leq r$ and $a_j=0$ for all $j \neq i$. Then, the CA is essentially a ``shift operator'' and again, studying $\HHH\SSS\Pfin \Xi (\AA)$ is simple.
    \item \label{can_add3} There are two non-zero coefficients $a_i$ and $a_j$, $i < j$ (we can take $i$ to be the smallest such coefficient and $j$ the largest one). Hence, $\AA \in \AFF_{r; i,j}^\F$ and Theorem \ref{thm:full_main_result} applies.
\end{enumerate}

Note that in the simple cases \ref{can_add1}.\ and \ref{can_add2}., the automata in $\HHH\SSS\Pfin \Xi (\AA)$ are not all affine. 

\chng{
\subsection{Generalizations} \label{subsection:generalizations}
In a closely related paper \cite{additive_cas_over_Zp_bottom}, the authors showed that certain automata are limited in terms of what they can simulate. All such automata can be interpreted as affine over $\F_p$, $p$ prime (as a matter of fact, they are canonical linear). Thus, studying the classes $\AFF^{\F_p}_{r;i,j}$, $-r \leq i < j \leq r$, seems like a natural generalization; this is one of the main reasons we focused on such automata in this paper.

However, it is relevant to consider other related classes of automata, such as abelian CAs or CAs affine over arbitrary finite fields. Below, we introduce abelian CAs and discuss how the results from this paper can be directly used to infer simulation limitations of such general CA classes.

\paragraph{CAs affine over arbitrary finite fields}
Let $\F$ be an arbitrary finite field. Then $\F = \F_{p^n}$ for some prime $p$ and $n \in \N$. We have already showed in Example \ref{ex:counterexample_finite_field} that if $n>1$ then CAs simulated by automata affine over $\F$ need not be affine over $\F$. However, it is an easy observation that $\F_{p^n}$ is a vector space over $\F_p$ and thus $\AFF_r^{\F_{p^n}} \subseteq \AFF_r^{\F_{p}}$ (and analogously $\LIN_r^{\F_{p^n}} \subseteq \LIN_r^{\F_{p}}$). This immediately yields the following result:
\begin{thm} \label{thm:full_main_result_arbitrary_field}
    For any $r \in \N$, $-r \leq i < j \leq r$, any prime $p$ and any $n \in \N$ it holds that:  
    \begin{align*}
    \HHH\SSS\Pfin \Xi(\AFF^{\F_{p^n}}_{r;i,j})&\subseteq \AFF^{\F_{p}}_{r;i,j},\\
    \HHH\SSS\Pfin \Xi(\LIN^{\F_{p^n}}_{r;i,j})&\subseteq \AFF^{\F_{p}}_{r;i,j}.
\end{align*}
\end{thm}

\paragraph{Abelian CAs}
One may notice that to obtain our main result in Theorem \ref{thm:full_main_result}, we only used the additive structure of vector spaces over $\F_{p}$, since multiplication by scalars from $\F_p$ can be obtained by repeated addition. This motivates the following definition:

\begin{defn} [affine abelian CA, abelian CA]
Let $(G, +, -, 0)$ be a finite abelian group and let $\A = (G^\Z, F)$ be a CA with local algebra $(G, f)_r$. We say that $\AA$ (or $\A$) is \emph{abelian} if $f: G^{2r+1} \rightarrow G$ is a homomorphism between the abelian groups. In such a case, we can write $f$ in the following form:
\begin{align*} 
    f(\xx_{-r}, \ldots, \xx_r) = f_{-r}(\xx_{-r}) + \cdots + f_r(\xx_r),
\end{align*}
where $f_i: G \rightarrow G$ is a group endomorphism for each $-r \leq i \leq r$. The class of all local algebras isomorphic to some abelian local algebra of a CA with radius $r$ is denoted as $\AB_r$.

In the more general case when $f$ has the form
\begin{align*} 
    f(\xx_{-r}, \ldots, \xx_r) = f_{-r}(\xx_{-r}) + \cdots + f_r(\xx_r) + \cc,
\end{align*}
with $f_i: G \rightarrow G$ a group endomorphism for each $-r \leq i \leq r$ and $\cc \in G$ a constant, we say $\AA$ (or $\A$) is \emph{affine abelian}. We denote the class of all local algebras isomorphic to some affine abelian local algebra of a CA with radius $r$ as $\AFAB_r$.

For $-r \leq i < j \leq r$ we again define classes $\AB_{r;i,j}$ and $\AFAB_{r;i,j}$ that contain only $(i,j)$-bi-permutive algebras.
\end{defn}

It is easy to see that the notion of an abelian CA generalizes linearity over finite fields; indeed, for any finite field $\F$ it holds that:
\begin{align*}
    \AFF^{\F}_{r}&\subseteq \AFAB_{r}, \\
    \LIN^{\F}_{r}&\subseteq \AB_{r}.  
\end{align*}

An interested reader can go through the proofs of Lemma \ref{lemma:ksi_affinity}, Proposition \ref{lemma:affine_subautomata}, Corollary \ref{cor:s(aff)=aff}, Proposition \ref{prop:s(add)=aff}, Lemma \ref{lemma:affine_quotient_automata}, and Corollary \ref{cor:h(aff)=aff} to see that they strightforwardly generalize to the case of abelian and affine abelian CAs. Thus, this immediately yields the following general result.

\begin{thm} \label{thm:full_main_abelian_result}
    For any $r \in \N$, $-r \leq i < j \leq r$, it holds that  
    \begin{align*}
    \HHH\SSS\Pfin \Xi(\AFAB_{r;i,j})&=\AFAB_{r;i,j}, \\
    \HHH\SSS\Pfin \Xi(\AB_{r;i,j})&=\AFAB_{r;i,j}. 
\end{align*}
\end{thm}

Compared to the main Theorem \ref{thm:full_main_result}, this result shows limitations of the simulation capacity of a broader class of automata. On the other hand, the limitation is less strict -- thus, for CAs affine over a finite field, it is better to apply Theorem \ref{thm:full_main_result}.

To complement Theorem \ref{thm:full_main_abelian_result}, recall Examples \ref{ex:counterexample_s} and \ref{ex:counterexample_h}: Clearly, any affine abelian CA with $q$ elements, $q$ prime, is in fact affine over $\F_q$. Thus, these two examples actually present affine abelian CAs violating the bijective condition with a sub-automaton or quotient automaton that is not affine abelian. 
}
\section{Concluding Remarks} \label{section:conclusion} 
We stress that an important part of the merit of this paper lies in formalizing the notion of CA simulation into algebraic language. This makes it possible to see new connections to well established fields of abstract algebra. Whereas the proofs provided in this paper do not rely on any sophisticated algebraic concepts, we remark that, as an example, Lemma \ref{lemma:affine_quotient_automata} and Corollary \ref{cor:h(aff)=aff} are a direct consequence of a deeper theorem by Smith \cite{malcev_varieties} and Gumm \cite{algebras_in_congruence_permutable_varieties} about abelian algebras with a Maltsev term.

We believe that the connection with abstract algebra can provide powerful tools for deriving a plethora of both negative and positive results regarding the simulation capacity of various CA classes in the future.

\section{Acknowledgements} We thank Jiří Tůma for supervising the whole process of this paper's creation as well as David Stanovský for his valuable insights. \chng{We further thank the reviewers for their insightful comments that helped us present our work in a broader context.}

Our work was supported by the Czech project AI$\&$Reasoning CZ.02.1.01/0.0/0.0/15\texttt{\char`_}003/0000466 and the European Regional Development Fund and by SVV-2020-260589.

\printbibliography

\end{document}